\documentclass[journal,12pt,onecolumn,draftclsnofoot,]{IEEEtran}
\usepackage{amssymb}
\usepackage{amsmath}	
\usepackage{amsthm}
\usepackage{verbatim}
\usepackage{rotating}
\usepackage{amsbsy}
\usepackage{algorithm}	
\usepackage{algpseudocode}
\usepackage{graphicx}	
\usepackage{epstopdf}
\usepackage{epsfig}
\usepackage{mathdots}
\usepackage{mathtools}
\usepackage{fixltx2e}
\usepackage{enumitem}
\usepackage{subcaption}
\usepackage{multicol}
\usepackage{setspace}
\newtheorem{theorem}{Theorem}

\newtheorem{lemma}{Lemma}
\theoremstyle{definition}
\newtheorem{definition}{Definition}

\usepackage{chngcntr}
\usepackage{apptools}
\AtAppendix{\counterwithin{proposition}{section}}

\makeatletter
\def\footnoterule{\relax%
  \kern-1pt
  \hbox to \columnwidth{\vrule width 0.5\columnwidth height 0.4pt\hfill}
  \kern4.6pt}
\makeatother

\usepackage{caption}
\usepackage{url}
\usepackage{breqn}
\usepackage{stmaryrd}
\usepackage[mathscr]{euscript}
\DeclareMathOperator{\tr}{tr}

\begin{document}
\title{Linear to multi-linear algebra and systems using tensors}
\author{Divyanshu~Pandey, Adithya Venugopal,
        Harry~Leib
\thanks{This work was supported by the Natural Sciences and Engineering Research Council of Canada (NSERC) titled ‘‘Tensor modulation for
space-time-frequency communication systems’’ under Grant RGPIN-2016-03647.

Author's email : divyanshu.pandey@mail.mcgill.ca, adithya.venugopal@mail.mcgill.ca, harry.leib@mcgill.ca}} 
\maketitle

\vspace*{-1.5cm}
\begin{abstract}
In the past few decades, multi-linear algebra also known as tensor algebra has been adapted and employed as a tool for various engineering applications. Recent developments in tensor algebra have indicated that several well-known concepts from Linear Algebra can be extended to a multi-linear setting with the help of a special form of tensor contracted product, known as the Einstein product. Thus, the tensor contracted product and its properties can be harnessed to define the notions of multi-linear system theory where the input, output signals, and the system are inherently multi-domain or multi-modal.  This paper provides an overview of tensor algebra tools which can be seen as an extension of Linear Algebra, at the same time highlighting the differences and advantages that the multi-linear setting brings forth. In particular, the notions of tensor inversion, tensor singular value and tensor eigenvalue decomposition using the Einstein product are explained. In addition, this paper also introduces the notion of contracted convolution in both discrete and continuous multi-linear system tensors. Tensor Network representation of various tensor operations is also presented. Also, application of tensor tools in developing transceiver schemes for multi-domain communication systems, with an example of MIMO CDMA system, is presented.  This paper provides a foundation for professionals whose research involves multi-domain or multi-modal signals and systems.
    
\end{abstract}


\begin{spacing}{1.63}
\section{Introduction}
Tensors are multi-way arrays that are indexed by multiple indices and the number of indices is called the \textit{order} of the tensor \cite{KoldaTensor}. Subsequently, matrices and vectors can be seen as order two and order one tensors respectively. Higher-order tensors are inherently capable of mathematically representing processes and systems with dependency on more than two variables. Hence tensors are widely employed for several applications in many engineering and science disciplines. Tensors were initially introduced for applications in Physics during the early nineteenth century \cite{PierreComon}. Later with the work of Tucker \cite{tucker64extension}, tensors were used in Psychometrics in the 1960s for extending two-way data analysis to higher-order datasets, and further in Chemometrics in the 1980s \cite{chemo1981, bro2006review}. The last few decades have witnessed a surge in their applications in areas such as  data mining \cite{DataMining,TensorsDataMining}, computer vision \cite{ComputerVision,ComputerVision2}, neuroscience \cite{Neuroscience}, machine learning \cite{NikosTensor}, signal processing \cite{Cichocki,NikosGianPFA,MMSEJournal}, multi-domain communications \cite{AdithyaPaper, MDPIpaper} and controls system theory \cite{MLTI1,MLTI2}. When appropriately employed, tensors can help in developing models that capture interactions between various parameters of multi-domain systems. Such tensor-based system representation can enhance the understanding of the mutual effects of various system domains. 

Given the wide scope of applications that tensors support, there have been many recent publications summarizing the essential topics in tensor algebra. One such primary reference is \cite{KoldaTensor} where the fundamental tensor decompositions such as Tucker, PARAFAC and their variants are discussed in great detail with applications. Another useful reference is \cite{PierreComon} which presents tensors as a mapping from one linear space to another, along with a discussion on tensor ranks. A more signal processing oriented outlook on tensors is considered in \cite{Cichocki}, including applications such as Big Data storage and Compressed sensing. A more recent and exhaustive tutorial style paper is \cite{NikosTensor} which presents a detailed overview of up-to-date tensor decomposition algorithms, computations and applications in machine learning. Similarly, \cite{TensorTut2021} presents such an overview with applications in multiple-input multiple-output (MIMO) wireless communications. Also, \cite{TensorBook2020} provides a detailed review of many tensor decompositions with a focus on the needs of Data Analytics community.
However, all these papers do not consider in particular the notions of tensor contracted product and contracted convolution, which are the crux of this paper. With the help of a specific form of contracted product, known as the Einstein product of tensors, various tensor decompositions and properties can be established which may be viewed as an intuitive and meaningful extension of the corresponding linear algebra concepts.

The most popular and widely used decompositions in the case of matrices are the singular value and eigenvalue decompositions. In order to consider their extensions to higher-order tensors, it is important to note there is no single generalization that preserves all the properties of the matrix case \cite{limsingular,BelzenSVD}. The most commonly used generalization of the matrix singular value decomposition is known as a Higher-order Singular Value Decomposition (HOSVD)  which is basically the same as Tucker decomposition for higher-order tensors \cite{LathauwerSVD}. Similarly, several definitions exist in the literature for tensor eigenvalues as a generalization of the matrix eigenvalues \cite{QiLEigen}. More recently, in order to solve a set of multi-linear equations using tensor inversion,  a specific notion of tensor singular value decomposition and eigenvalue decomposition was introduced in \cite{TamonTensorInversion}, which generalizes the matrix SVD and EVD to tensors through a fixed transformation of the tensors into matrices. The authors in \cite{TamonTensorInversion} establish the equivalence between the Einstein product of tensors and the matrix product of the transformed tensors, thereby proving that a tensor group endowed with the Einstein product is structurally similar or isomorphic to a general linear group of matrices. The notion of equivalence between the Einstein product of tensors and the corresponding matrix product of the transformed tensors is very crucial and relevant as it helps in developing many tools and concepts from matrix theory such as matrix inverse, ranks, and determinants for tensors. Hence as a follow-up to \cite{TamonTensorInversion}, several other works explored different notions of linear algebra which can be extended to multi-linear algebra using the Einstein product \cite{LuEVD,SunMoore,TensorDet, huang2021numerical2, huang2021numerical, huang2018iterative, wang2018iterative}.

The purpose of this paper is twofold. First, we intend to present an overview of tensor algebra concepts developed in the past decade using the Einstein product. Since there is a natural way of extending the linear algebra concepts to tensors, the idea here is to present a summary of the most commonly used and relevant concepts which can equip the reader with tools to define and prove other properties more specific to their intended applications. Secondly, this paper introduces the notion of contracted convolutions for both discrete and continuous system tensors. The theory of linear time invariant (LTI) systems has been an indispensable tool in various engineering applications such as Communication Systems, and Controls. Now with the evolution of these subjects to multi-domain communication systems and multi-linear systems theory, there is a need to better understand the classical topics in a multi-domain setting. This paper intends to provide such tools through a tutorial style presentation of the subject matter leading to a mechanism to develop more tools needed for research and applications in any multi-domain/multi-dimensional/multi-modal/multi-linear setting.   

The organization of this paper is as follows: in section \ref{Sec2}, we present basic tensor definitions and operations, including the concept of signal tensors and contracted convolutions. In section \ref{Sec3}, we present the Tensor Networks representation of various tensor operations. Section \ref{Sec4} presents some tensor decompositions based on the Einstein product. Section \ref{Sec5}  defines the notions of multi-linear system tensors and discusses their stability in both time and frequency domains. It also includes a detailed discussion on the application of tensors to multi-linear system representation with an example of MIMO CDMA system. The paper is concluded in section \ref{Sec6}.

\section{Fundamentals of tensors and notation}\label{Sec2}
A tensor is a multi-way array whose elements are indexed by three or more indices. Each index may correspond to a different domain, dimension, or mode of the quantity being represented by the array. The order of the tensor is the number of such indices or domains or dimensions or modes. A vector is often referred to as a tensor of order-1, a matrix as a tensor of order-2 and tensors of order greater than 2 are known as higher-order tensors. 

\subsection{Notations}
In this paper we use lowercase underline fonts to represent vectors, e.g. $\underline{\text{x}}$, uppercase fonts to represent matrices, e.g. $\text{X}$ and uppercase calligraphic fonts to represent tensors, e.g. $\mathscr{X}$. The individual elements of a tensor are denoted by the indices in subscript, e.g. the $(i_1,i_2,i_3)$th element of a third-order tensor $\mathscr{X}$ is denoted by ${\mathscr{X}}_{i_1,i_2,i_3} $. A colon in subscript for a mode corresponds to every element of that mode corresponding to fixed other modes. For instance, ${\mathscr{X}}_{:,i_2,i_3}$ denotes every element of tensor $\mathscr{X}$ corresponding to $i_2$th second and $i_3$th third mode.  The $n$th element in a sequence is denoted by a superscript in parentheses, e.g. ${\mathscr{A}}^{(n)}$ denotes the $n^{th}$ tensor in a sequence of tensors. We use $\mathbb{C}$ and $\mathbb{C}_k$ to denote a set of complex numbers and a set of complex numbers which are a function of $k$, respectively.

\subsection{Definitions and tensor operations}

\begin{definition}{\textbf{Tensor Linear Space }}:
The set of all tensors of size $I_1 \times \dots \times I_K$ over $\mathbb{C}$ forms a linear space, denoted as $\mathbb{T}_{I_1,\dots,I_K}(\mathbb{C})$. For $\mathscr{A},\mathscr{B} \in \mathbb{T}_{I_1,\dots,I_K}(\mathbb{C})$ and $\alpha \in \mathbb{C}$, the sum $\mathscr{A}+\mathscr{B} = \mathscr{C} \in \mathbb{T}_{I_1,\dots,I_K}(\mathbb{C})$ where $\mathscr{C}_{i_1,\dots,i_k}=\mathscr{A}_{i_1,\dots,i_k}+\mathscr{B}_{i_1,\dots,i_k}$, and scalar multiplication $\alpha \cdot \mathscr{A}= \mathscr{D} \in \mathbb{T}_{I_1,\dots,I_K}(\mathbb{C})$ where $\mathscr{D}_{i_1,\dots,i_k}=\alpha \mathscr{A}_{i_1,\dots,i_k}$.
\end{definition}

\begin{definition}{{\bf Fiber :}}
Fiber is defined by fixing every index in a tensor but one. A matrix column is a mode-1 fiber and matrix row is a mode-2 fiber. A third order tensor has similarly column (mode-1), row (mode-2) and tube (mode-3) fibers \cite{KoldaTensor}.
\end{definition}

\begin{definition}{{\bf Slices :}}
Slices are two dimensional sections of a tensors defined by fixing all but two indices. 
\end{definition}

\begin{definition}{{\bf Norm :}}
The $p-$norm of an order $N$ tensor $\mathscr{X} \in \mathbb{C}^{I_{1} \times I_{2} \times \dots \times I_{N} }$ is defined as 
\begin{equation}
\Vert \mathscr{X} \Vert_p = \Big(\sum_{i_{1}=1}^{I_{1}} \sum_{i_{2}=1}^{I_{2}} ...... \sum_{i_{N}=1}^{I_{N}} \mid \mathscr{X}_{i_{1},i_{2},....,i_{N}}\mid^{p}\Big)^{1/p}  
\end{equation}

Subsequently, the $2-$norm or the Frobenius norm of $\mathscr{X}$ is defined as the square root of the sum of the square of absolute values of all its elements : \\
\begin{equation} \label{Norm}
\Vert \mathscr{X} \Vert_2 = \sqrt{\sum_{i_{1}=1}^{I_{1}} \sum_{i_{2}=1}^{I_{2}} ...... \sum_{i_{N}=1}^{I_{N}} \mid \mathscr{X}_{i_{1},i_{2},....,i_{N}}\mid^{2}}  
\end{equation}
Also, the $1-$norm and $\infty-$norm of a tensor are defined as :
\begin{align}
\Vert \mathscr{X} \Vert_1 = \sum_{i_1,\dots,i_N} \mid \mathscr{X}_{i_{1},i_{2},....,i_{N}}\mid
\\
\Vert \mathscr{X} \Vert_{\infty} = \max_{i_1,\dots,i_N} \mid \mathscr{X}_{i_{1},i_{2},....,i_{N}}\mid 
\end{align}
\end{definition} 

\begin{definition}{{\bf Kronecker Product of Matrices :}} 
The Kronecker product of two matrices $\text{A}$ of size $I \times J$ and $\text{B}$ of size $K \times L$, denoted by $\text{A} \otimes \text{B}$, is a matrix of size $(IK) \times (JL)$ and is defined as
\begin{equation}
\text{A} \otimes \text{B} = \begin{bmatrix}
    \text{A}_{1,1}\text{B}&\text{A}_{1,2}\text{B}&\dots&\text{A}_{1,J}\text{B}\\
    \text{A}_{2,1}\text{B}&\text{A}_{2,2}\text{B}&\dots&\text{A}_{2,J}\text{B}\\
    \vdots&\vdots&\ddots&\vdots\\
    \text{A}_{I,1}\text{B}&\text{A}_{I,2}\text{B}&\dots&\text{A}_{I,J}\text{B}\\
  \end{bmatrix} 
\end{equation}
\end{definition}

\begin{definition}{\textbf{Matricization Transformation }}:
Let us denote the linear space of $P \times Q$ matrices over $\mathbb{C}$ as $\mathbb{M}_{P,Q}(\mathbb{C})$. For an order $K=N+M$ tensor $\mathscr{A} \in \mathbb{C}^{I_1\times \dots \times I_N \times J_1 \times \dots \times J_M}$, the transformation $f_{I_1,\dots,I_N|J_1,\dots,J_M} : \mathbb{T}_{I_1,\dots,I_N,J_1,\dots,J_M}(\mathbb{C}) \Rightarrow \mathbb{M}_{I_1 \cdot I_2\cdots I_{N-1} \cdot I_N,J_1 \cdot J_2\cdots J_{M-1}\cdot J_M} (\mathbb{C})$ with $f_{I_1,\dots,I_N|J_1,\dots,J_M}(\mathscr{A})=\text{A}$ is defined component-wise as \cite{TamonTensorInversion} :
\begin{equation} \label{Transform}
\mathscr{A}_{i_1,i_2,\dots,i_N,j_1,j_2,\dots, j_M} \xrightarrow{f_{I_1,\dots,I_N|J_1,\dots,J_M}}\text{A}_{i_1+\sum_{k=2}^{N}(i_k-1)\prod_{l=1}^{k-1}I_l , j_1+\sum_{k=2}^{M}(j_k-1)\prod_{l=1}^{k-1}J_l }
\end{equation}
\end{definition}
This transformation is referred to as matricization, or matrix unfolding by partitioning the indices into two disjoint subsets \cite{BaderTensor}. The vectorization operation as defined in \cite{DeLvectorize} can be seen as a specific case of (\ref{Transform}) by using $J_1=\dots=J_M=1$.  The bar in subscript of $f_{I_1,\dots,I_N|J_1,\dots,J_M}$ represents the partitioning after $N$ modes of an $N+M$ order tensor where first $N$ modes correspond to the rows of the representing matrix, and the last $M$ modes correspond to the columns of the representing matrix. This mapping is bijective \cite{TensorDet}, and it preserves addition and scalar multiplication operations i.e., for $\mathscr{A},\mathscr{B} \in \mathbb{T}_{I_1,\dots,I_N,J_1,\dots,J_M}(\mathbb{C})$ and any scalar $\alpha \in \mathbb{C}$, we have $f_{I_1,\dots,I_N|J_1,\dots,J_M}(\mathscr{A}+\mathscr{B})=f_{I_1,\dots,I_N|J_1,\dots,J_M}(\mathscr{A})+f_{I_1,\dots,I_N|J_1,\dots,J_M}(\mathscr{B})$ and $f_{I_1,\dots,I_N|J_1,\dots,J_M}(\alpha\mathscr{A})=\alpha f_{I_1,\dots,I_N|J_1,\dots,J_M}(\mathscr{A})$. Hence the linear spaces $ \mathbb{T}_{I_1,\dots,I_N,J_1,\dots,J_M}(\mathbb{C}) $ and $ \mathbb{M}_{I_1 \cdot I_2\cdots I_{N-1} \cdot I_N,J_1 \cdot J_2\cdots J_{M-1}\cdot J_M} (\mathbb{C})$ are isomorphic and the transformation $f_{I_1,\dots,I_N|J_1,\dots,J_M}$ is an isomorphism between the linear spaces. For a matrix, the transformation (\ref{Transform}) does no change when $N=M=1$, creates a column vector when $N=2,M=0$ and a row vector when $N=0,M=2$.

\subsubsection{Tensor Products}
Tensors have multiple modes, hence a product between two tensors can be defined in various ways. In this section, we present definitions of the most commonly used tensor products. 
\begin{definition}{{\bf Tensor Contracted product }}\cite{BaderTensor}:
Consider two tensors $\mathscr{X} \in \mathbb{C}^{I_{1} \times I_{2} \times \dots \times I_{M} \times J_{1} \times J_{2} \times \dots \times J_{N} }$ and $\mathscr{Y} \in \mathbb{C}^{I_{1} \times I_{2} \times \dots \times I_{M} \times K_{1} \times K_{2} \times \dots \times K_{P} }$. We can multiply both tensors along their common $M$ modes and the resulting tensor, $\mathscr{Z} \in \mathbb{C}^{J_{1} \times J_{2} \times \dots \times J_{N} \times K_{1} \times K_{2} \times \dots \times K_{P} } $ is given by :
\begin{equation}\label{ContractedProduct}
\mathscr{Z} = \{ \mathscr{X},\mathscr{Y} \}_{\{1,\dots,M;1,\dots,M\}}
\end{equation}
where
\begin{equation}
\mathscr{Z}_{j_1,\dots,j_N,k_1,\dots,k_P} = \sum_{i_1 = 1}^{I_1}\dots \sum_{i_M = 1}^{I_M} \mathscr{X}_{i_{1}, \dots, i_{M}, j_{1} , \dots, j_{N}}\mathscr{Y}_{i_{1} , \dots, i_{M} , k_{1} , \dots ,k_{P}} 
\end{equation}

\end{definition}
It is important to note that the modes to be contracted need not be consecutive, however the size of the corresponding dimensions must be equal. For example, tensors $\mathscr{A} \in \mathbb{C}^{K \times L \times M \times N}$ and $\mathscr{B} \in \mathbb{C}^{K \times M \times Q \times R}$ can be contracted along the first and third mode of $\mathscr{A}$ and first and second mode of $\mathscr{B}$ as $\mathscr{C}= \{ \mathscr{A},\mathscr{B} \}_{\{1,3;1,2\}}$ where $\mathscr{C} \in \mathbb{C}^{L \times N \times Q \times R}$. Matrix multiplication between $\text{A} \in \mathbb{C}^{I \times J}$ and $\text{B} \in \mathbb{C}^{J \times K}$ can be seen as a specific case of the contracted product as $\text{A}\cdot \text{B} = \{ \text{A},\text{B} \}_{\{2;1\}}$ where $\cdot$ represents usual matrix multiplication. Several other tensor products can be defined as specific cases of contracted product. One such commonly used tensor product is the Einstein product where the modes to be contracted are at a fixed location as defined next.

\begin{definition}{{\bf Einstein product }}:
The Einstein product between tensors $\mathscr{A} \in \mathbb{C}^{I_1 \times \dots \times I_P \times K_1 \dots \times K_N}$ and $\mathscr{B} \in \mathbb{C}^{K_1 \times \dots \times K_N \times J_{1} \dots \times J_M}$ is defined as a contraction between their $N$ common modes, denoted by $*_N$, as \cite{TamonTensorInversion} :
\begin{equation} \label{EinsteinProduct}
(\mathscr{A}*_N \mathscr{B})_{i_1,\dots,i_P,j_{1},\dots,j_{M}} = \sum_{k_1,\dots,k_N}\mathscr{A}_{i_1,i_2,\dots,i_P,k_1,\dots,k_N}\mathscr{B}_{k_1,\dots k_N,j_{1},j_{2},\dots,j_M}
\end{equation}
\end{definition}
In Einstein product, contraction is over $N$ consecutive modes and can also be written using the more general notation with contracted modes in subscript (say for sixth order tensors $\mathscr{H}, \mathscr{S} \in \mathbb{C}^{I \times J \times K \times I \times J \times K}$)  :
\begin{equation}
(\mathscr{H}*_3 \mathscr{S})_{i,j,k,\hat{i},\hat{j},\hat{k}}=\sum_{u=1}^{I}\sum_{v=1}^{J}\sum_{w=1}^{K}\mathscr{H}_{i,j,k,u,v,w}\mathscr{S}_{u,v,w,\hat{i},\hat{j},\hat{k}} = (\{ \mathscr{H}, \mathscr{S} \}_{\{4,5,6;1,2,3\}})_{i,j,k,\hat{i},\hat{j},\hat{k}}
\end{equation}
Note that one can define Einstein product for several specific mode orderings. For instance, in \cite{MLTI1} Einstein product is defined as contraction over $N$ alternate modes and not consecutive modes. However, that would not change the concepts presented here on so far as we remain consistent with the definition.  

\begin{definition}{{\bf Inner Product :}}
The inner product of two tensors $\mathscr{X} \in \mathbb{C}^{I_{1} \times I_{2} \times \dots \times I_{N}}$ and $\mathscr{Y}\in \mathbb{C}^{I_{1} \times I_{2} \times \dots \times I_{N}}$ of the same order $N$ with all the dimensions of same length is given by : \\
\begin{equation} \label{InnerProduct}
\langle \mathscr{X}, \mathscr{Y} \rangle = \sum_{i_{1}=1}^{I_{1}} \sum_{i_{2}=1}^{I_{2}} \dots \sum_{i_{N}=1}^{I_{N}} \mathscr{X}_{i_{1},i_{2},\dots,i_{N}} \mathscr{Y}_{i_{1},i_{2},\dots,i_{N}} 
\end{equation}
It can also be seen as the Einstein product of tensors where contraction is along all the dimensions, i.e.  $\langle \mathscr{X}, \mathscr{Y} \rangle = \mathscr{X} *_N \mathscr{Y} = \mathscr{Y} *_N \mathscr{X}$. 
\end{definition} 

\begin{definition}{{\bf Outer Product :}}
Consider two tensors $\mathscr{X} \in \mathbb{C}^{I_{1} \times I_{2} \times \dots \times I_{N} }$ and $\mathscr{Y} \in \mathbb{C}^{J_{1} \times J_{2} \times \dots \times J_{M} }$ of order $N$ and $M$ respectively. The outer product between $\mathscr{X}$ and $\mathscr{Y}$ denoted by $\mathscr{X} \circ \mathscr{Y} $ is given by a tensor of size $ {I_{1} \times I_{2} \times \dots \times I_{N} \times J_{1} \times J_{2} \times \dots \times J_{M}} $ with individual elements as : \\
\begin{equation} \label{OuterProduct}
( \mathscr{X} \circ \mathscr{Y} )_{i_{1},i_{2},\dots,i_{N},j_{1},j_{2},\dots,j_{M}} = \mathscr{X}_{i_{1},i_{2},\dots,i_{N}} \mathscr{Y}_{j_{1},j_{2},\dots,j_{M}}  
\end{equation}
It can also be seen as the special case of Einstein product of tensors in equation (\ref{EinsteinProduct}) with $N=0$. 
\end{definition} 

\begin{definition}{\textbf{n-mode product}}:
The $n$-mode product of a tensor $\mathscr{A} \in \mathbb{C}^{I_{1} \times I_{2} \times \dots \times I_{N} }$ with a matrix $\text{U} \in \mathbb{C}^{J \times I_{n}}$ is denoted by $\mathscr{A} {\times}_{n} \text{U}$ and is defined as \cite{LathauwerSVD} :
\begin{equation} \label{TensorTimesMatrix}
(\mathscr{A} {\times}_{n} \text{U})_{i_1,i_2,\dots,i_{n-1},j,i_{n+1},\dots,i_N} = \sum_{i_n = 1}^{I_n} \mathscr{X}_{i_1,i_2,\dots,i_N} \text{U}_{j,i_n} 
\end{equation}
Each mode-$n$ fiber is multiplied by the matrix $\text{U}$. The result of $n$-mode product is a tensor of same order but with a new $n$th mode of size $J$. The resulting tensor is of size $I_1 \times I_2 \times \dots \times I_{n-1} \times J \times I_{n+1} \times \dots I_N$.\\
\end{definition}

\begin{definition}{\textbf{Square tensors }}: A tensor $\mathscr{A} \in \mathbb{C}^{I_1 \times \dots \times I_N \times J_1 \times \dots \times J_M}$ is called a square tensor if $N=M$ and $I_k=J_k$ for $k=1,\dots,N$ \cite{TensorDet}.
\end{definition}
   
For square tensors $\mathscr{A},\mathscr{B}$ of size $I \times J \times I \times J$, it was shown in \cite{TamonTensorInversion} that $f_{I,J|I,J}(\mathscr{A}*_2 \mathscr{B}) = f_{I,J|I,J}(\mathscr{A}) \cdot f_{I,J|I,J}(\mathscr{B})$ where $\cdot$ refers to the usual matrix multiplication. It was further generalized to square or non-square tensors of any order and size in \cite{wang2018iterative} as the following Lemma:
\begin{lemma}\label{TransformPropertyLemma}
For tensors $\mathscr{A} \in \mathbb{C}^{I_1 \times \dots \times I_N \times J_1 \times \dots \times J_M}$ and $\mathscr{B} \in \mathbb{C}^{J_1 \times \dots \times J_M \times K_1 \times \dots \times K_P}$ under the transformation from (\ref{Transform}), the following holds:
\begin{equation}\label{TranformProperty}
f_{I_1,\dots,I_N|K_1,\dots,K_P}(\mathscr{A}*_M \mathscr{B})=f_{I_1,\dots,I_N|J_1,\dots,J_M}(\mathscr{A})\cdot f_{J_1,\dots,J_M|K_1,\dots,K_P}(\mathscr{B}) 
\end{equation} 
\end{lemma}

\begin{definition}{\textbf{Pseudo-diagonal Tensors }}:
Any tensor $\mathscr{D} \in \mathbb{C}^{I_1 \times \dots \times I_N \times J_1 \times \dots \times J_M}$ of order $N+M$ is called pseudo-diagonal if its transformation $\text{D} = f_{I_1,\dots,I_N|J_1,\dots,J_M}(\mathscr{D})$ yields a diagonal matrix such that $\text{D}_{i,j}$ is non-zero only when $i=j$. 
\end{definition}
Since the transformation (\ref{Transform}) is bijective, we can say that a diagonal matrix $\text{D} \in \mathbb{C}^{I \times J}$ under inverse transformation $f^{-1}_{I_1, \dots ,I_N | J_1, \dots ,J_M}(\text{D})$ will yield a pseudo-diagonal tensor where $I=I_1\cdots I_N$ and $J=J_1 \cdots J_M$.  A square tensor $\mathscr{D} \in \mathbb{C}^{I_1 \times \dots \times I_N \times I_1 \times \dots \times I_N}$ is pseudo-diagonal if all its entries $\mathscr{D}_{i_1,\dots,i_N,j_1,\dots,j_N}$ are zero except when $i_1=j_1,i_2=j_2,\dots,i_N=j_N$. In \cite{TamonTensorInversion,SunMoore} such a tensor is termed as diagonal tensor, and in \cite{MLTI1} it is termed as U-diagonal, but we tend to call it pseudo-diagonal for our purpose of discussion, so as to make a clear distinction from the diagonal tensor definition more widely found in literature which states that a diagonal tensor is one where entries $\mathscr{D}_{i_1,\dots ,i_N}$ are zero except when $i_1=i_2=\dots =i_N$ \cite{KoldaTensor}. This can be seen as a more strict diagonal condition as non-zero elements exist only when all the modes have same index whereas in a pseudo-diagonal tensor, say of order $2N$, elements are non-zero when every $i$th and $(i+N)$th mode have same index for $i=1,\dots,N$. An illustration of order 4 tensor showing the difference between diagonal and pseudo-diagonal structures can be found in \cite{MDPIpaper}. For a matrix, which has just two modes, the diagonal and pseudo-diagonal structures are the same. 
Note that pseudo-diagonality is defined with respect to partition after $N$ modes. For instance, if we refer to  a third order tensor as pseudo-diagonal, then it is important to specify whether it is pseudo-diagonal with respect to partition after the first mode or the second mode. So to avoid overload of notation in this paper whenever we write a tensor explicitly as $N+M$ or $2N$ order tensor and call it pseudo-diagonal, then pseudo-diagonality is with respect to partition after $N$ modes. 


\begin{definition}{\textbf{Pseudo-Triangular Tensor:}}
A tensor $\mathscr{A} \in \mathbb{C}^{I_1 \times \dots \times I_N \times I_1 \times \dots \times I_N}$ is defined to be pseudo-lower triangular if
\begin{equation}
\mathscr{A}_{i_1, \dots ,i_N,i_1', \dots ,i_N'} = \begin{cases}0 \quad \text{if } (i_1'+\sum\limits_{k=2}^{N}(i_k'-1)\prod\limits_{l=1}^{k-1}I_l) \geq (i_1+\sum\limits_{k=2}^{N}(i_k-1)\prod\limits_{l=1}^{k-1}I_l) \\
a_{i_1, \dots ,i_N,i_1', \dots ,i_N'} \quad \text{otherwise}
\end{cases}
\end{equation}
where $a_{i_1, \dots ,i_N,i_1', \dots ,i_N'}$ are arbitrary scalars. Similarly, the tensor is said to be pseudo-upper triangular if 
\begin{equation}
\mathscr{A}_{i_1, \dots ,i_N,i_1', \dots ,i_N'} = \begin{cases}0 \quad \text{if } (i_1'+\sum\limits_{k=2}^{N}(i_k'-1)\prod\limits_{l=1}^{k-1}I_l) \leq (i_1+\sum\limits_{k=2}^{N}(i_k-1)\prod\limits_{l=1}^{k-1}I_l) \\
a_{i_1, \dots ,i_N,i_1', \dots ,i_N'} \quad \text{otherwise}
\end{cases}
\end{equation} 
An illustration of an upper triangular tensor of size $J_1 \times J_2 \times I_1 \times I_2$ with $I_1 = I_2 = J_1 = J_2 = 3$ is presented in Figure \ref{Pseudo_up} and its pseudo-upper triangular elements highlighted in gray along with its pseudo-diagonal elements shown in black. A similar illustration of a lower triangular tensor can be found in \cite{AdithyaPaper}. It can be readily seen that a lower triangular tensor becomes a lower triangular matrix under the tensor to matrix transformation defined in \eqref{Transform} and a pseudo-upper triangular tensor becomes an upper triangular matrix.
\end{definition}

\begin{figure}
\center
\includegraphics[scale=1.0]{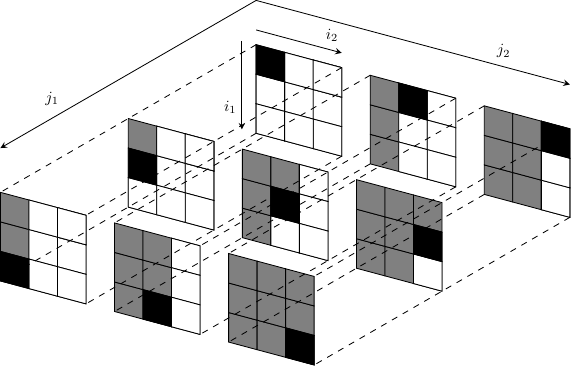}
\caption{Pseudo-upper triangular Tensor \label{Pseudo_up}}
\end{figure}

\begin{definition}{\textbf{Identity Tensor }}:
An identity tensor, $\mathscr{I}_N \in \mathbb{C}^{I_1 \times \dots \times I_N \times I_1 \times \dots \times I_N }$ is a pseudo-diagonal tensor of order $2N$ such that for any tensor $\mathscr{A} \in \mathbb{C}^{I_1 \times \dots \times I_N \times I_1 \times \dots \times I_N }$, we have $\mathscr{A}*_N\mathscr{I}_N=\mathscr{I}_N *_N\mathscr{A}=\mathscr{A} $ and in which all non-zero entries are 1, i.e.
\begin{equation}\label{IdentityTensor}
(\mathscr{I}_{N})_{i_1,i_2,\dots,i_N,j_1,j_2,\dots,j_N}=\prod_{k=1}^{N}\delta_{i_kj_k},\,\,\,\,\, \text{where }  \delta_{pq}= 
\begin{cases}
    1,& p = q\\
    0,& p \neq q
\end{cases}
\end{equation}
\end{definition}

\subsubsection{Transpose, Hermitian and Inverse of a Tensor}
The transpose of a matrix is a permutation of its two indices corresponding to rows and columns. Since elements of a higher-order tensor are indexed by multiple indices, there are several permutations of such indices and hence there can be multiple ways to write the transpose or Hermitian of a tensor. Such permutation dependent transpose of a tensor is defined in \cite{pan2014tensor}.

Assume the set $S_N = \{1,2, \dots ,N \}$ and $\sigma$ is a permutation of $S_N$.  We denote $\sigma(j) = i_j\ \text{for } j = 1,2, \dots , N$ where $\{i_1,i_2, \dots , i_N\} = \{1,2, \dots, N\} = S_N$. Since $S_N$ is a finite set with $N$ elements, it has $N!$ different permutations. Hence, discounting the identity permutation $\sigma(j) = [1,2, \dots, N]$, there are $N!-1$ different transposes for a tensor with $N$ dimensions or modes. For a tensor $\mathscr{A} \in \mathbb{C}^{I_1 \times \dots \times I_N}$ we define its transpose associated with a certain permutation $\sigma$ as $\mathscr{A}^{T \sigma} \in \mathbb{C}^{I_{\sigma(1)} \times \dots \times I_{\sigma(N)}}$ with entries 
\begin{equation}
\mathscr{A}^{T \sigma}_{i_{\sigma (1)},i_{\sigma (2)}, \dots , i_{\sigma (N)}} = \mathscr{A}_{i_1,i_2, \dots , i_N}.
\end{equation}
Similarly, the Hermitian of a tensor $\mathscr{A} \in \mathbb{C}^{I_1 \times \dots \times I_N}$ associated with a permutation $\sigma$ is defined as the conjugate of its transpose and is denoted as $\mathscr{A}^{H \sigma} = (\mathscr{A}^{T \sigma})^* \in \mathbb{C}^{I_{\sigma(1)} \times \dots \times I_{\sigma(N)}}$ with entries
\begin{equation}
\mathscr{A}^{H \sigma}_{i_{\sigma (1)},i_{\sigma (2)}, \dots , i_{\sigma (N)}} = (\mathscr{A}^{T \sigma}_{i_{\sigma (1)},i_{\sigma (2)}, \dots , i_{\sigma (N)}})^* = (\mathscr{A}_{i_1,i_2, \dots , i_N})^*.
\end{equation}
For example, a transpose of a third-order tensor $\mathscr{X} \in \mathbb{C}^{I_1 \times I_2 \times I_3}$ such that its third mode is transposed with the first can be written as $\mathscr{X}^{T \sigma}$ where $\sigma = [3,2,1]$ with components $\mathscr{X}^{T \sigma}_{i_3,i_2,i_1} = \mathscr{X}_{i_1,i_2,i_3}$. 
For two tensors $\mathscr{A} \in \mathbb{C}^{I_1 \times \dots \times I_N}$ and $\mathscr{B} \in \mathbb{C}^{I_1 \times \dots \times I_N}$ we have \cite{pan2014tensor}
\begin{equation}
\langle \mathscr{A},\mathscr{B} \rangle = \langle \mathscr{A}^{T\sigma},\mathscr{B}^{T\sigma} \rangle.
\end{equation}

Consider a tensor $\mathscr{Y} \in \mathbb{C}^{I_1 \times \dots \times I_N \times J_1 \times \dots \times J_M}$ with a transposition such that the final $M$ modes are swapped with the first $N$ modes can be represented by a permutation function $\sigma = [(N+1),\dots(N+M), 1, \ \dots \ N]$ such that $\mathscr{Y}^{T \sigma}_{j_1, \dots, j_M, i_1, \dots i_N} = \mathscr{Y}_{i_1, \dots, i_N,j_1, \dots , j_M}$. Since we will use tensors to define system theory elements with fixed order $M$ output and order $N$ input, the most often encountered case of transpose or Hermitian in this paper would be after $N$ modes of an $N+M$ or $2N$ tensor, i.e. $\sigma = [(N+1),\dots(N+M), 1, \ \dots \ N]$. Henceforth, in such a case we drop the superscript $\sigma$ for ease of representation and represent such a transpose by  $\mathscr{Y}^{T}$ and its conjugate by $\mathscr{Y}^H$.

Further, a square tensor $\mathscr{U} \in \mathbb{C}^{I_1 \times \dots \times I_N \times I_1 \times \dots \times I_N}$ is called a \textit{unitary} tensor if $\mathscr{U}^H*_N \mathscr{U} = \mathscr{U}*_N \mathscr{U}^H = \mathscr{I}_N$.

The tensor $\mathscr{A}^{-1} \in \mathbb{C}^{I_1 \times \dots \times I_N \times I_1 \times \dots \times I_N}$ is an \textit{inverse} of a square tensor of same size, $\mathscr{A} \in \mathbb{C}^{I_1 \times \dots \times I_N \times I_1 \times \dots \times I_N}$ if $\mathscr{A} *_N \mathscr{A}^{-1} = \mathscr{A}^{-1} *_N \mathscr{A} =\mathscr{I}_N$ \cite{TensorDet}. The inverse of a tensor exists if its transformation $f_{I_1,\dots,I_N|I_1,\dots,I_N}(\mathscr{A})$ is invertible \cite{TamonTensorInversion}. Several algorithms using the Einstein product such as Higher-order Bi-conjugate Gradient method  \cite{TamonTensorInversion} or Newton's method \cite{FICCpaper} can be used to find tensor inverse without relying on actually transforming the tensor into a matrix.   

As a generalization of the matrix Moore-Penrose inverse, the \textit{Moore-Penrose inverse} of a tensor $\mathscr{A} \in \mathbb{C}^{I_1 \times \dots \times I_N \times J_1 \times \dots \times J_N}$ is defined as a tensor $\mathscr{A}^{+} \in \mathbb{C}^{J_1 \times \dots \times J_N \times I_1 \times \dots \times I_N}$ that satisfies \cite{SunMoore, MPIExnMishra}:
\begin{align}
&{\mathscr{A}}*_N{\mathscr{A}}^+*_N{\mathscr{A}} = {\mathscr{A}},\notag\\
&{\mathscr{A}}^+*_N{\mathscr{A}}*_N{\mathscr{A}}^+ = {\mathscr{A}}^+,\notag\\
&({\mathscr{A}}*_N{\mathscr{A}}^+)^H = {\mathscr{A}}*_N{\mathscr{A}}^+,\notag\\
&({\mathscr{A}}^+*_N{\mathscr{A}})^H = {\mathscr{A}}^+*_N{\mathscr{A}}.\notag
\end{align}
For a tensor $\mathscr{A} \in \mathbb{C}^{I_1 \times \dots \times I_N \times J_1 \times \times \dots \times J_N}$, the Moore-Penrose inverse always exist and is unique \cite{SunMoore}.   

Based on the definition of tensor inverse, Hermitian, and the Einstein product, several tensor algebra relations and properties can be derived. Here we present a few properties which are often used. Along similar lines, several more properties can be derived. 
\begin{enumerate}[label=(\alph*)]
\item \textit{Associativity} : For tensors $\mathscr{A} \in \mathbb{C}^{I_1 \times \dots \times I_P \times J_1 \times \dots \times J_N}$, $\mathscr{B} \in \mathbb{C}^{J_1 \times \dots \times J_N \times K_1 \times \dots \times K_M}$ and $\mathscr{C} \in \mathbb{C}^{K_1 \times \dots \times K_M \times T_1 \times \dots \times T_Q}$, we have
\begin{equation}\label{associative}
(\mathscr{A} *_N \mathscr{B}) *_M \mathscr{C} = \mathscr{A} *_N (\mathscr{B} *_M \mathscr{C})
\end{equation}
\begin{proof} 
\begin{align}
&(({\mathscr{A}} *_N {\mathscr{B}}) *_M {\mathscr{C}})_{i_1,\dots,i_P,t_1,\dots,t_Q} = \sum_{k_1,\dots,k_M}\Big(\sum_{j_1,\dots,j_N}{\mathscr{A}}_{i_1,\dots,i_P,j_1,\dots,j_N}{\mathscr{B}}_{j_1,\dots,j_N,k_1,\dots,k_M} \Big){\mathscr{C}}_{k_1,\dots,k_M,t_1,\dots,t_Q} \nonumber \\
&=  \sum_{j_1,\dots,j_N}{\mathscr{A}}_{i_1,\dots,i_P,j_1,\dots,j_N}\sum_{k_1,\dots,k_M}{\mathscr{B}}_{j_1,\dots,j_N,k_1,\dots,k_M}{\mathscr{C}}_{k_1,\dots,k_M,t_1,\dots,t_Q} = ({\mathscr{A}} *_N ({\mathscr{B}} *_M {\mathscr{C}}))_{i_1,\dots,i_P,t_1,\dots,t_Q} \nonumber
\end{align}

\end{proof}

\item \textit{Commutativity} : Einstein product is not commutative in general. However for the specific case where the product is taken over all the $N$ modes of one of the tensors, say for tensors $\mathscr{A} \in \mathbb{C}^{I_1 \times \dots \times I_P \times J_1 \times \dots \times J_N}$ and $\mathscr{B} \in \mathbb{C}^{J_1 \times \dots \times J_N}$, the following can be established :
\begin{equation}\label{commutative}
\mathscr{A} *_N \mathscr{B} = \mathscr{B} *_N \mathscr{A}^T
\end{equation}
\begin{proof}
\begin{align}
({\mathscr{A}} *_N {\mathscr{B}})_{i_1,\dots,i_P} &=\sum_{j_1,\dots,j_N}{\mathscr{A}}_{i_1,\dots,i_P,j_1,\dots,j_N}{\mathscr{B}}_{j_1,\dots,j_N} \nonumber \\
& = \sum_{j_1,\dots,j_N}{\mathscr{B}}_{j_1,\dots,j_N}{\mathscr{A}}_{j_1,\dots,j_N,i_1,\dots,i_P}^T =  ({\mathscr{B}} *_N {\mathscr{A}}^T)_{i_1,\dots,i_P} \nonumber
\end{align}
\end{proof}

\item \textit{Distributivity} : For tensors, ${\mathscr{A}},{\mathscr{B}} \in \mathbb{C}^{I_1 \times \dots \times I_P \times J_1 \times \dots \times J_N}$ and ${\mathscr{C}} \in \mathbb{C}^{J_1 \times \dots \times J_N \times K_1 \times \dots \times K_M}$, we have :
\begin{equation}\label{distributive}
({\mathscr{A}}+{\mathscr{B}})*_N {\mathscr{C}} = ({\mathscr{A}}*_N {\mathscr{C}})+({\mathscr{B}}*_N {\mathscr{C}})   
\end{equation}
\begin{proof}
\begin{equation*}
\begin{aligned}
&(({\mathscr{A}}+{\mathscr{B}})*_N {\mathscr{C}})_{i_1,\dots,i_P,k_1,\dots,k_M} = \sum_{j_1,\dots,j_N}({\mathscr{A}}_{i_1,\dots,i_P,j_1,\dots,j_N} + {\mathscr{B}}_{i_1,\dots,i_P,j_1,\dots,j_N}){\mathscr{C}}_{j_1,\dots,j_N,k_1,\dots,k_M}\\
&=(\sum_{j_1,\dots,j_N}{\mathscr{A}}_{i_1,\dots,i_P,j_1,\dots,j_N}{\mathscr{C}}_{j_1,\dots,j_N,k_1,\dots,k_M}) + (\sum_{j_1,\dots,j_N}{\mathscr{B}}_{i_1,\dots,i_P,j_1,\dots,j_N}{\mathscr{C}}_{j_1,\dots,j_N,k_1,\dots,k_M}) \\
& = ({\mathscr{A}}*_N {\mathscr{C}})_{i_1,\dots,i_P,k_1,\dots,k_M}+({\mathscr{B}}*_N {\mathscr{C}})_{i_1,\dots,i_P,k_1,\dots,k_M}
\end{aligned}
\end{equation*}
\end{proof}
\item For tensors ${\mathscr{A}} \in \mathbb{C}^{I_1 \times \dots \times I_M \times J_1 \times \dots \times J_N}$ and ${\mathscr{B}} \in \mathbb{C}^{J_1 \times \dots \times J_N \times K_1 \times \dots \times K_P}$, we have :
\begin{equation}
({\mathscr{A}} *_N {\mathscr{B}})^H = {\mathscr{B}}^H *_N {\mathscr{A}}^H 
\end{equation}
\begin{proof}
\begin{equation*}
\begin{aligned}
({\mathscr{A}} *_N {\mathscr{B}})_{k_1,\dots,k_P,i_1,\dots,i_M}^H &= \Big(\sum_{j_1,\dots,j_N}{\mathscr{A}}_{i_1,\dots,i_M,j_1,\dots,j_N}{\mathscr{B}}_{j_1,\dots,j_N,k_1,\dots,k_P}\Big)^*\\
&= \sum_{j_1,\dots,j_N}{\mathscr{A}}_{i_1,\dots,i_M,j_1,\dots,j_N}^*{\mathscr{B}}_{j_1,\dots,j_N,k_1,\dots,k_P}^* \\
&= \sum_{j_1,\dots,j_N}{\mathscr{B}}_{k_1,\dots,k_P,j_1,\dots,j_N}^H {\mathscr{A}}_{j_1,\dots,j_N,i_1,\dots,i_M}^H \\
&= ({\mathscr{B}}^H *_N {\mathscr{A}}^H )_{k_1,\dots,k_P,i_1,\dots,i_M}
\end{aligned}
\end{equation*}
\end{proof}

\item For square tensors ${\mathscr{A}}$ and ${\mathscr{B}} \in \mathbb{C}^{I_1 \times \dots \times I_N \times I_1 \times \dots \times I_N}$, we have :
\begin{equation}
({\mathscr{A}} *_N {\mathscr{B}})^{-1} = {\mathscr{B}}^{-1} *_N {\mathscr{A}}^{-1} 
\end{equation}  
\begin{proof}
Let $ {\mathscr{C}} = {\mathscr{B}}^{-1} *_N {\mathscr{A}}^{-1}$. If ${\mathscr{C}}$ is an inverse of $({\mathscr{A}} *_N {\mathscr{B}})$ then as per the definition of tensor inverse, it should satisfy $({\mathscr{A}} *_N {\mathscr{B}}) *_N {\mathscr{C}} = {\mathscr{C}} *_N ({\mathscr{A}} *_N {\mathscr{B}}) = {\mathscr{I}}_N$.
\begin{equation*}
\begin{aligned}
({\mathscr{A}} *_N {\mathscr{B}}) *_N ({\mathscr{B}}^{-1} *_N {\mathscr{A}}^{-1} ) &= {\mathscr{A}} *_N ( {\mathscr{B}} *_N {\mathscr{B}}^{-1}) *_N {\mathscr{A}}^{-1} \,\,\,\,\,\,\, \text{(from associativity)} \\
&= ({\mathscr{A}} *_N {\mathscr{I}}_N) *_N {\mathscr{A}}^{-1} = {\mathscr{A}} *_N {\mathscr{A}}^{-1} = {\mathscr{I}}_N\\
\text{Similarly, } ({\mathscr{B}}^{-1} *_N {\mathscr{A}}^{-1} ) *_N ({\mathscr{A}} *_N {\mathscr{B}}) &= {\mathscr{B}}^{-1} *_N ( {\mathscr{A}}^{-1} *_N {\mathscr{A}}) *_N {\mathscr{B}} = {\mathscr{I}}_N\\ 
\end{aligned}
\end{equation*}
\end{proof}
\end{enumerate}

\subsubsection{Function Tensors}
A \textit{function tensor} ${\mathscr{A}}(x) \in \mathbb{C}_x^{I_1 \times \dots \times I_N}$ is an order $N$ tensor whose components are functions of $x$. Using a third order function tensor as an example, each component of ${\mathscr{A}}(x)$ is written as ${\mathscr{A}}_{i,j,k}(x)$. If $x$ takes discrete values, we represent the function tensor using square bracket notation as $\mathscr{A}[x]$.

A generalization of the function tensor would be the \textit{multivariate function tensor} ${\mathscr{A}}(x_1,\dots,x_p) \in \mathbb{C}^{I_1 \times \dots \times I_N}_{x_1,\dots,x_p}$, which is an order $N$ tensor whose components are functions of the continuous variables $x_1,\dots,x_p$. If the variables take discrete values, we denote the function tensor as $\mathscr{A}[x_1,\cdots,x_p]$. Using the same example of a third order tensor, each component can be written as ${\mathscr{A}}_{i,j,k}(x_1,x_2,\dots,x_p)$. 

A linear system is often expressed as $\text{A}\underline{\text{x}}=\underline{\text{b}}$ where $\text{A} \in \mathbb{C}^{M \times N}$ is a matrix operating upon the vector $\underline{\text{x}} \in \mathbb{C}^{N}$ to produce another vector $\underline{\text{b}} \in \mathbb{C}^M$ \cite{TamonTensorInversion}. Essentially, the matrix defines a linear operator $\mathcal{L} : \mathbb{C}^{N} \rightarrow \mathbb{C}^M$ between two vector linear spaces $\mathbb{C}^N$ and $\mathbb{C}^M$. A multi-linear system can be thus defined as a linear operator between two tensor linear spaces $\mathbb{C}^{I_1 \times \dots \times I_N}$ and $\mathbb{C}^{J_1 \times \dots \times J_M}$, i.e. $\mathcal{ML} : \mathbb{C}^{I_1 \times \dots \times I_N} \rightarrow \mathbb{C}^{J_1 \times \dots \times J_M}$. Multi-linear systems model several phenomenon in various science and engineering applications. However often in literature, a multi-linear system is degenerated into a linear system by mapping the tensor linear space $\mathbb{C}^{I_1 \times \dots \times I_N}$ into a vector linear space $\mathbb{C}^{I_1\cdots I_N}$ through vectorization. The vectorization process allows one to use tools from linear algebra for convenience but also leads to a representation where distinction between different modes of the system is lost. Thus possible hidden patterns, structures, and correlations cannot be explicitly identified in the vectorized tensor entities. With the help of tensor contracted product, one can develop signals and system representation without having to rely on vectorization, at the same time extending tools from linear to a multi-linear setting intuitively. 

\subsection{Discrete Time Signal Tensors}

A \textit{discrete time signal tensor} ${\mathscr{X}}[n] \in \mathbb{C}^{I_1 \times \dots \times I_N}_n$ is a function tensor whose components are functions of the sampled time index $n$. A discrete tensor signal can also be called a tensor sequence indexed by $n$.

A multi-linear time invariant \textit{discrete system tensor} is an order $N+M$ tensor sequence $\mathscr{H}[k] \in \mathbb{C}_k^{J_1 \times \dots \times J_M \times I_1 \times \dots \times I_N}$ that couples an input tensor sequence $\mathscr{X}[k] \in \mathbb{C}_k^{I_1 \times \dots \times I_N}$ of order $N$ with an output tensor sequence $\mathscr{Y}[k] \in \mathbb{C}_k^{J_1 \times \dots \times J_M}$ of order $M$ through a discrete contracted convolution defined as:
\begin{equation}\label{dfeeq1}
\mathscr{Y}[k] = \sum_{n}\{\mathscr{H}[n],\mathscr{X}[k-n]\}_{\{M+1,\dots,M+N ; 1,\dots,N\}}.
\end{equation}
Most often the ordering of the modes while defining such system tensors is fixed, where the system tensor contracts over all the input modes. Hence for a more compact notation, we can define the contracted convolution using the Einstein product as:
\begin{equation}\label{dfeeq2}
\mathscr{Y}[k] = \sum_{n} \mathscr{H}[n] *_N \mathscr{X}[k-n]. 
\end{equation}
In scalar signals and systems notations, a convolution between two functions is often represented using an asterisk $(*)$. However, to make a distinction with the Einstein product notation which also uses the asterisk symbol, we denote the contracted convolution using the notation $\bullet_N$, i.e.
\begin{equation}\label{dfeeq3}
\mathscr{Y}[k] = \mathscr{H}[k] \bullet_N \mathscr{X}[k] = \sum_{n} \mathscr{H}[n] *_N \mathscr{X}[k-n]. 
\end{equation}
The complex frequency domain representation of discrete signal tensors can be given using the z-transform of the signal tensors, as discussed next.

\begin{definition}{\textbf{z-transform of a Discrete Tensor Sequence:}}
The z-transform of $\mathscr{X}[n] \in \mathbb{C}^{I_1 \times \dots I_N}_n$ denoted by $\breve{\mathscr{X}}(z)=\mathcal{Z}(\mathscr{X}[k]) \in \mathbb{C}^{I_1 \times \dots \times I_N}_z$ is a tensor of the z-transform of its components defined as
\begin{equation}
\breve{\mathscr{X}}(z) = \mathcal{Z}(\mathscr{X}[k]) = \sum\limits_{n}\mathscr{X}[n]z^{-n}
\end{equation}
with components $\breve{\mathscr{X}}_{i_1, \dots , i_N}(z) = \sum\limits_{n}\mathscr{X}_{i_1, \dots , i_N}[n]z^{-n}$.
\end{definition}
The discrete time Fourier transform denoted by $\bar{\mathscr{X}}(\omega) = \mathcal{F}(\mathscr{X}[k])$ of a tensor sequence can be found by substituting $z=e^{j\omega}$ in its z-transform as
\begin{equation}
\bar{\mathscr{X}}(\omega) =\breve{\mathscr{X}}(z)\bigg\rvert_{z=e^{j \omega}} = \sum\limits_{n}\mathscr{X}[n]z^{-n} \bigg\rvert_{z=e^{j \omega}} = \sum\limits_{n}\mathscr{X}[n]e^{-j\omega n}. 
\end{equation}

Taking the z-transform of \eqref{dfeeq3} we get
\begin{align}
\breve{\mathscr{Y}}(z) &= \sum_{k}\mathscr{Y}[k]z^{-k} = \sum_{k}\bigg(\sum_{n}\mathscr{H}[n]*_N \mathscr{X}[k-n]\bigg)z^{-k} \notag\\
&= \sum_{k}\bigg(\sum_{n}\mathscr{H}[n] *_N \mathscr{X}[k-n]\bigg)z^{n-k}z^{-n} = \sum_{n} \mathscr{H}[n] *_N \bigg(\sum_{k}\mathscr{X}[k-n]z^{n-k}\bigg) z^{-n} \notag\\
&= \sum_{n}\mathscr{H}[n]z^{-n} *_N \breve{\mathscr{X}}(z) = \bigg(\sum_{n}\mathscr{H}[n]z^{-n}\bigg) *_N \breve{\mathscr{X}}(z) \notag\\\label{dtraneq1}
&= \breve{\mathscr{H}}(z) *_N \breve{\mathscr{X}}(z).
\end{align}
which shows that the discrete contracted convolution between two tensors in the time domain as given by \eqref{dfeeq3} leads to the Einstein product between the tensors in the z-domain.

\subsection{Continuous Time Signal Tensors}\label{ContSigTen}
A \textit{continuous time signal tensor} ${\mathscr{X}}(t) \in \mathbb{C}^{I_1 \times \dots \times I_N}_t$ is a function tensor whose components are functions of the continuous time variable $t$. 

A multi-linear time invariant \textit{continuous system tensor} is an order $N+M$ tensor $\mathscr{H}(t) \in \mathbb{C}_t^{J_1 \times \dots \times J_M \times I_1 \times \dots \times I_N}$ that couples an order $N$ input continuous tensor signal $\mathscr{X}(t) \in \mathbb{C}_t^{I_1 \times \dots \times I_N}$ with an order $M$ output tensor signal $\mathscr{Y}(t) \in \mathbb{C}_t^{J_1 \times \dots \times J_M}$ through a continuous contracted convolution defined as:
\begin{equation}\label{cfeeq1}
\mathscr{Y}(t) = \int \{\mathscr{H}(u),\mathscr{X}(t-u)\}_{\{M+1,\dots,M+N ; 1,\dots,N\}} du.
\end{equation}
In cases where the mode sequence is fixed, similar to the discrete case, we can define a more compact notation using the Einstein product as:
\begin{equation}\label{cfeeq2}
\mathscr{Y}(t) = \mathscr{H}(t) \bullet_N \mathscr{X}(t) = \int \mathscr{H}(u) *_N \mathscr{X}(t-u) du. 
\end{equation}
The frequency domain representation of continuous signal tensors can be given using the Fourier transform of the signal tensor as defined next.

\begin{definition}{\textbf{ Fourier transform}}:
The Fourier transform of $\mathscr{X}(t) \in \mathbb{C}^{I_1 \times \dots \times I_N}_t$ denoted by $\bar{\mathscr{X}}(\omega)=\mathcal{F}(\mathscr{X}(t))$ is a tensor of the Fourier transform of its components defined as
\begin{equation}
\bar{\mathscr{X}}(\omega) = \mathcal{F}(\mathscr{X}(\omega)) = \int \mathscr{X}(t)e^{-j\omega t}dt
\end{equation}
with components $\bar{\mathscr{X}}_{i_1, \dots , i_N}(\omega) = \int \mathscr{X}_{i_1, \dots , i_N}(t)e^{-i \omega t} dt$.
\end{definition}
Using similar line of derivation as for \eqref{dtraneq1}, it can be shown that \eqref{cfeeq2} can be written in frequency domain as $\bar{\mathscr{Y}}(\omega)= \bar{\mathscr{H}}(\omega) *_N \bar{\mathscr{X}}(\omega)$.

\section{Tensor Networks}\label{Sec3}
Tensor Network (TN) diagrams are a graphical way of illustrating tensor operations \cite{cichocki2014era}. A TN diagram uses a node to represent a tensor and each outgoing edge from a node represents a mode of the tensor. As such, a vector can be represented through a node with a single edge, a matrix through a node with double edges, and an order $N$ tensor through a node with $N$ edges. This is illustrated in Figure \ref{VecMatTen}. Any form of tensor contraction can be visually presented through a TN diagram.
\begin{figure}
\center
\includegraphics[scale=0.8]{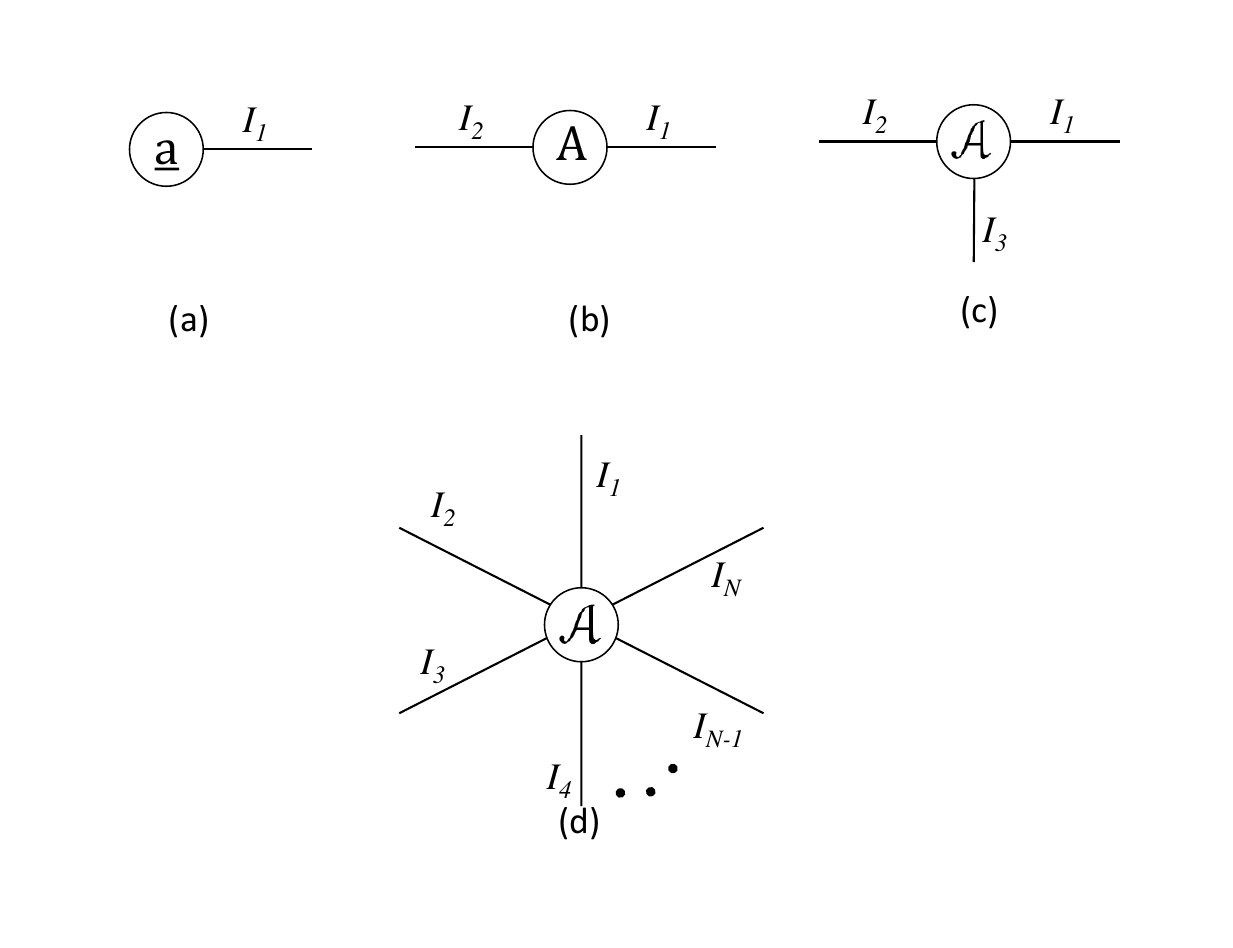}
\caption{TN diagram representation of (a) vector of size $I_1$, (b) matrix of size $I_1 \times I_2$, (c) order 3 tensor of size $I_1 \times I_2 \times I_3$, (d) order $N$ tensor of size $I_1 \times \dots \times I_N$ \label{VecMatTen}}
\end{figure}

\subsection{Illustration of Contraction Products}
A contraction between two modes of a tensor is represented in TN by connecting the edges corresponding to the modes which are to be contracted. Hence the number of free edges represent the order of the resulting tensor. As such any contracted product can be illustrated through a TN diagram, and a few examples are shown in Figure \ref{CP_TN_Fig}. In Figure \ref{CP_TN_Fig}(a), the mode-$n$ product of a tensor $\mathscr{A} \in \mathbb{C}^{I_1 \times \dots \times I_N}$ with $\text{U} \in \mathbb{C}^{J \times I_n}$, i.e. $\mathscr{A} \times_n \text{U}$ from \eqref{TensorTimesMatrix} is depicted where the $n$th edge of $\mathscr{A}$ is connected with the second edge of $\text{U}$ to represent the contraction of these modes of same dimension. Figure \ref{CP_TN_Fig}(b) shows the inner product between two third order tensors $\mathscr{A},\mathscr{B} \in \mathbb{C}^{I \times J \times K}$ where all the edges of both the tensors are connected. Since there is no free edge remaining, the result is a scalar. In Figure \ref{CP_TN_Fig}(c), a fifth order tensor $\mathscr{A} \in \mathbb{C}^{I \times J \times K \times L \times M}$ contracts with a fourth order tensor $\mathscr{B} \in \mathbb{C}^{J \times P \times L \times N}$ along its two common modes as $\{\mathscr{A},\mathscr{B}\}_{\{2,4;1,3\}}$. The resulting tensor is an order 5 tensor as there are total 5 free edges in the diagram. Finally, Figure \ref{CP_TN_Fig}(d) shows the Einstein product between tensors from \eqref{EinsteinProduct} where the common $N$ modes are connected and we have $P+M$ free edges.    
\begin{figure}
\center
\includegraphics[scale=0.5]{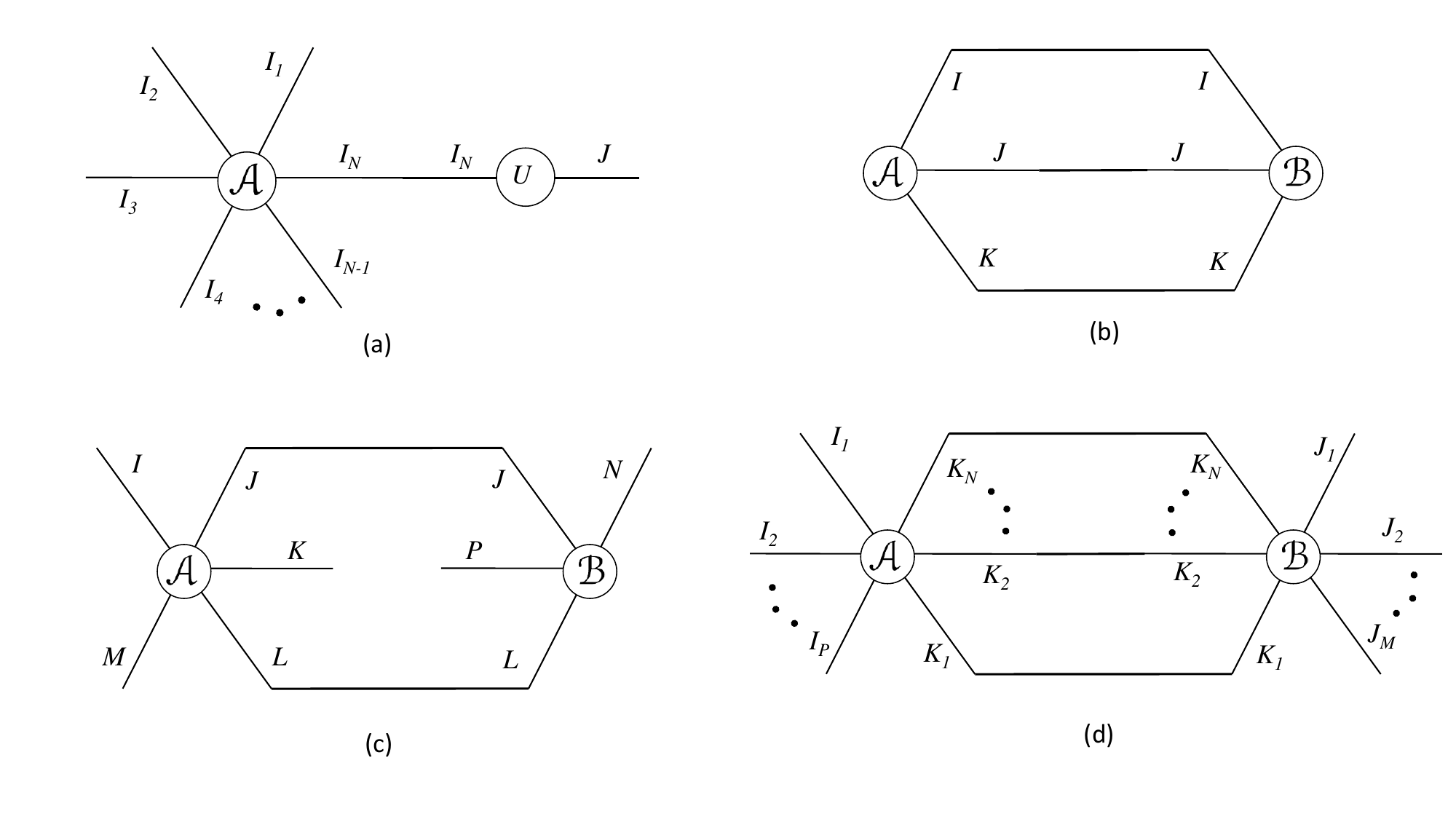}
\caption{TN representation of Contracted Product (a) mode-$n$ product between $\mathscr{A} \in \mathbb{C}^{I_1 \times \dots \times I_N}$ and $\text{U} \in \mathbb{C}^{J \times I_n}$ (b) between tensor $\mathscr{A},\mathscr{B} \in \mathbb{C}^{I \times J \times K}$ over all the three modes (inner product) (c) between $\mathscr{A} \in \mathbb{C}^{I \times J \times K \times L \times M}$ and $\mathscr{B} \in \mathbb{C}^{J \times P \times L \times N}$ as $\{\mathscr{A},\mathscr{B}\}_{\{2,4;1,3\}}$, (d) Einstein product between tensors $\mathscr{A} \in \mathbb{C}^{I_1 \times \dots \times I_P \times K_1 \times \dots \times K_N}$ and $\mathscr{B} \in \mathbb{C}^{K_1 \times \dots \times K_N \times J_1 \times \dots \times J_M}$. \label{CP_TN_Fig}}
\end{figure}

\subsection{Illustration of Contracted Convolutions}
A contracted convolution is an operation between tensor functions. A function tensor in a TN diagram is represented by a node which is a function of a variable. To represent a function tensor in a TN we use a rectangular node rather than a circular node. For a given value of the time index $k$, the contracted convolution from \eqref{dfeeq3} can be depicted using a TN diagram as shown in Figure \ref{ContConv_Fig}. Note that each $\mathscr{Y}[k]$ is calculated by computing Einstein product for all the values of $n$, hence the TN diagram contains a sequence of Einstein product representations between $\mathscr{H}[n]$ and $\mathscr{X}[k-n]$. We suggest a compact representation of the contracted convolution similar to contracted product, where we connect the edges of the function tensor using a dotted line to represent a contracted convolution as shown in Figure \ref{ContConv_Fig2}. This creates a distinction between the two representations. If the corresponding edges are connected via solid lines it represents contracted product, and if they are connected via dotted lines it represents contracted convolution.

Very often, TN diagrams are used to represent tensor operations as it provides a better visual understanding, and thereby aids in developing algorithms to compute tensor operations by making use of elements from graph theory and data structures. Furthermore, a TN diagram can also be used to illustrate how a tensor is formed from several other component tensors. Hence, most tensor decompositions studied in literature are often represented using a TN. In the next section, we discuss some tensor decompositions.

\begin{figure}
\center
\includegraphics[scale=0.5]{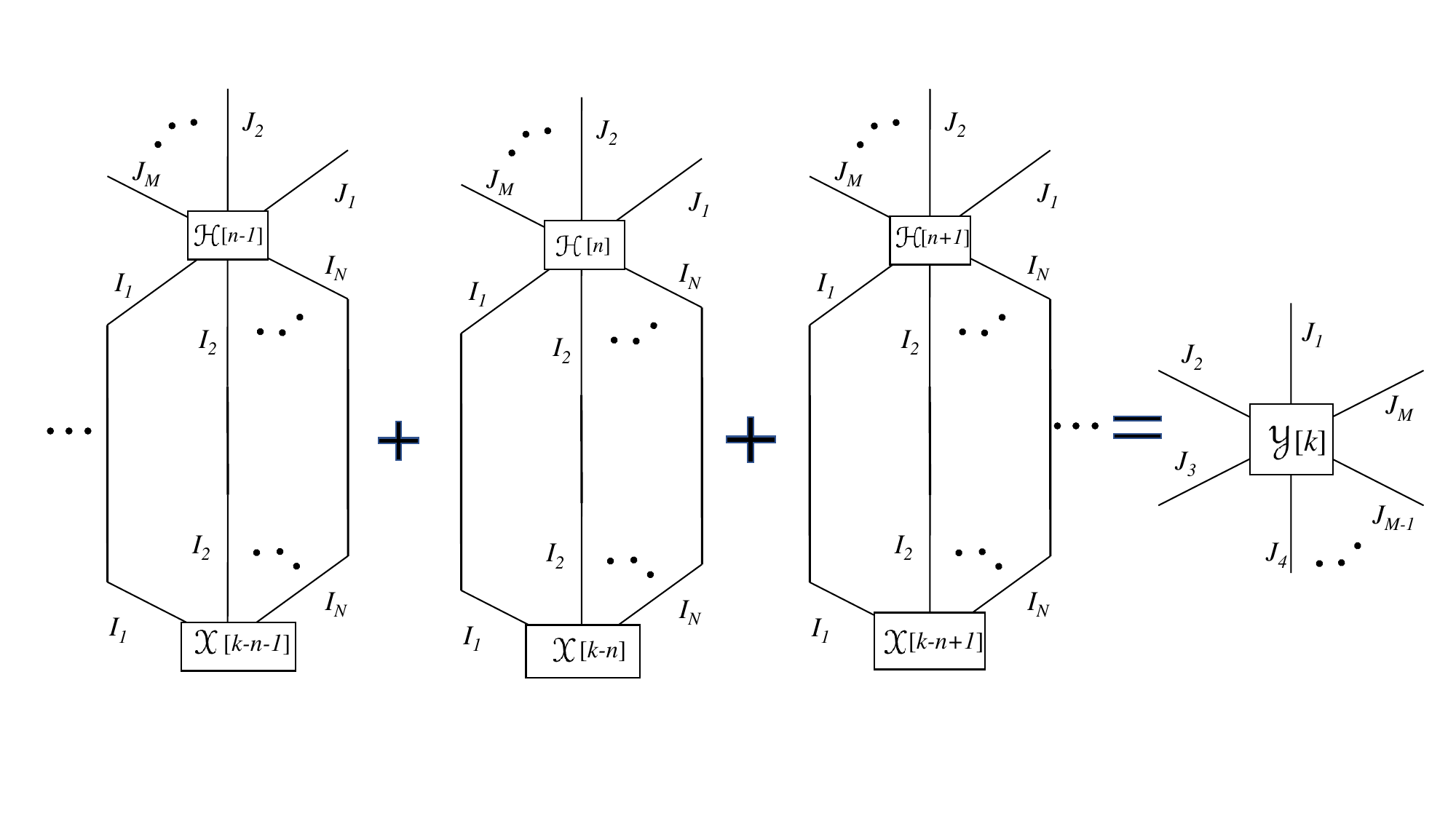}
\caption{TN representation of Contracted Convolution from \eqref{dfeeq3}. \label{ContConv_Fig}}
\end{figure}

\begin{figure}
\center
\includegraphics[scale=0.5]{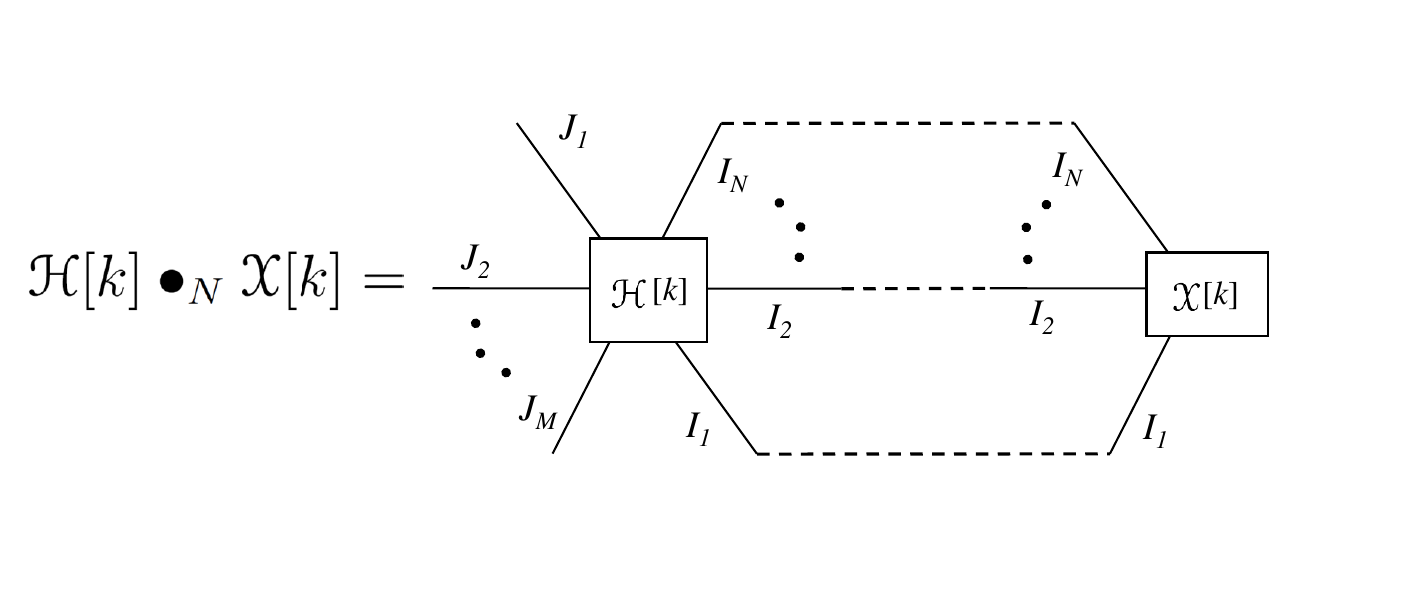}
\caption{Compact TN representation of Contracted Convolution from \eqref{dfeeq3}. \label{ContConv_Fig2}}
\end{figure}

\section{Tensor Decompositions}\label{Sec4}

Several Tensor decompositions such as the Tucker decomposition, Canonical Polyadic (CP) or the Parallel Factor (PARAFAC) decomposition, Tensor Train decomposition and many more have been extensively studied in the literature \cite{KoldaTensor, NikosTensor, TensorBook2020}. However, in the past decade with the help of Einstein product and its properties, a generalization of matrix SVD and EVD has been proposed in the literature which has found applications in solving multi-linear system of equations and systems theory \cite{MLTI2,TamonTensorInversion, LuEVD}. In this section, we present some such decompositions using the Einstein product of tensors.
 
\subsection{Tensor Singular Value Decomposition (SVD)} \label{tensorSVDsection}
Tucker decomposition of a tensor can be seen as a higher order SVD \cite{LathauwerSVD} and has found many applications particularly in extracting low rank structures in higher dimensional data \cite{zhang2018tensor}. A more specific version of tensor SVD is explored in \cite{TamonTensorInversion} as a tool for finding tensor inversion and solving multi-linear systems. Note that \cite{TamonTensorInversion} presents SVD for square tensors only. The idea of SVD from \cite{TamonTensorInversion} is further generalized for any even order tensor in \cite{SunMoore}. However, it can be further extended for any arbitrary order and size of tensor. We present a tensor SVD theorem here for any tensor of order $N+M$.
\begin{theorem}
For a tensor, $\mathscr{A} \in \mathbb{C}^{I_1 \times \dots \times I_N \times J_1 \times \dots \times J_M }$, the SVD of $\mathscr{A}$ has the form :
\begin{equation}\label{SVDTensor}
\mathscr{A} = \mathscr{U}*_N \mathscr{D} *_M \mathscr{V}^H
\end{equation}
where $\mathscr{U} \in \mathbb{C}^{I_1 \times \dots \times I_N \times I_1 \times \dots \times I_N}$ and $\mathscr{V} \in \mathbb{C}^{J_1 \times \dots \times J_M \times J_1 \times \dots \times J_M}$ are unitary tensors and $\mathscr{D} \in \mathbb{C}^{I_1 \times \dots \times I_N \times J_1 \times \dots \times J_M}$ is a pseudo-diagonal tensor whose non-zero values are the singular values of $\mathscr{A}$.  
\end{theorem} 
\begin{proof}
For tensors $\mathscr{A} \in \mathbb{C}^{I_1 \times \dots \times I_N \times J_1 \times \dots \times J_M}$ and $\mathscr{B} \in \mathbb{C}^{J_1 \times \dots \times J_M\times K_1 \times \dots \times K_P}$, from (\ref{TranformProperty}), we get :
\begin{equation}\label{invofA1}
\mathscr{A}*_M \mathscr{B}=f_{I_1,\dots,I_N|K_1,\dots,K_P}^{-1}[f_{I_1,\dots,I_N|J_1,\dots,J_M}(\mathscr{A})\cdot f_{J_1,\dots,J_M|K_1,\dots,K_P}(\mathscr{B})] 
\end{equation}
If $\text{A} \in \mathbb{C}^{I_1 I_2 \cdots I_N \times J_1 J_2 \cdots J_M}$ and $\text{B} \in \mathbb{C}^{J_1 J_2 \cdots J_M \times K_1 K_2 \cdots K_P}$ are transformed matrices from $\mathscr{A}$ and $\mathscr{B}$ respectively, then substituting $f_{I_1,\dots,I_N|J_1,\dots,J_M}(\mathscr{A})=\text{A}$ and $f_{J_1,\dots,J_M|K_1,\dots,K_P}(\mathscr{B}) = \text{B}$ in (\ref{invofA1}) gives us 
\begin{equation}\label{invfA}
f_{I_1,\dots,I_N|K_1,\dots,K_P}^{-1}(\text{A}\cdot \text{B})=\mathscr{A}*_M \mathscr{B} =f_{I_1,\dots,I_N|J_1,\dots,J_M}^{-1}(\text{A}) *_M f_{J_1,\dots,J_M|K_1,\dots,K_P}^{-1}(\text{B})
\end{equation}
Hence if $\text{A}=\text{U} \cdot \text{D} \cdot \text{V}^H$ (obtained from matrix SVD), then based on (\ref{invfA}), for an order $N+M$ tensor $\mathscr{A} \in \mathbb{C}^{I_1 \times \dots \times I_N \times J_1 \times \dots \times J_M}$, we have :
\begin{align}
\mathscr{A}&=f_{I_1,\dots,I_N|J_1,\dots,J_M}^{-1}(\text{A}) = f_{I_1,\dots,I_N|J_1,\dots,J_M}^{-1}(\text{U} \cdot \text{D} \cdot \text{V}^H) \nonumber \\  
&=f_{I_1,\dots,I_N|I_1,\dots,I_N}^{-1}(\text{U}) *_N f_{I_1,\dots,I_N|J_1,\dots,J_M}^{-1}(\text{D}) *_M f_{J_1,\dots,J_M|J_1,\dots,J_M}^{-1}(\text{V}^H) = \mathscr{U}*_N \mathscr{D} *_M \mathscr{V}^H   
\end{align}

\end{proof}
Note that for a tensor of order $N+M$, we will get different SVDs for different values of $N$ and $M$, but for a given $N$ and $M$, the SVD is unique which depends on the matrix SVD of $f_{I_1,\dots,I_N|J_1,\dots,J_M}(\mathscr{A})=\text{A}$. A proof of this theorem for $2N$ order tensors with $N=M$ using transformation defined in \eqref{Transform} is provided in \cite{TamonTensorInversion}. 
This SVD can be seen as a specific case of Tucker decomposition by expressing the unitary tensors in terms of the factor matrices obtained through Tucker decomposition. Let's consider an example of a fourth-order tensor. For a tensor $\mathscr{A} \in \mathbb{C}^{I_1 \times I_2 \times K_1 \times K_2}$, the Tucker decomposition has the form:
\begin{equation}\label{tuckerdecomp}
\mathscr{A} = \mathscr{D} {\times}_{1} \text{B}^{(1)} {\times}_{2} \text{B}^{(2)} {\times}_{3} \text{B}^{(3)} {\times}_{4} \text{B}^{(4)} 
\end{equation}
where $\text{B}^{(i)}$ are factor matrices along all four modes of the tensor and $\times_n$ denotes the $n$-mode product. Now, \eqref{tuckerdecomp} can be written in matrix form as follows \cite{LuEVD}:
\begin{equation}
\text{A}=\underbrace{(\text{B}^{(1)} \otimes \text{B}^{(2)})}_{\text{U}}\cdot \text{D}\cdot \underbrace{(\text{B}^{(3)T} \otimes \text{B}^{(4)T})}_{\text{V}^H}.
\end{equation}
Now using the transformation from \eqref{Transform}, we can map the elements of matrix $\text{U}$ to a tensor $\mathscr{U}$ as:
\begin{equation}
\mathscr{U}_{i,j,k,l}=\text{U}_{i+(j-1)I_1,k+(l-1)K_1}.
\end{equation}
Also since $\text{U}=(\text{B}^{(1)} \otimes \text{B}^{(2)})$ and Kronecker product when written element-wise, can be expressed as \cite{Cichocki}:
\begin{equation}
\begin{aligned}
\text{U}_{i+(j-1)I_1,k+(l-1)K_1} = \text{B}^{(1)}_{j,l} \cdot \text{B}^{(2)}_{i,k} \\
\Rightarrow \mathscr{U}_{i,j,k,l} = \text{B}^{(1)}_{j,l} \cdot \text{B}^{(2)}_{i,k}.
\end{aligned}
\end{equation}
This relation can be seen as the unitary tensor $\mathscr{U}$ being the outer product of matrices $\text{B}^{(1)}$ and $\text{B}^{(2)}$ \cite{TamonTensorInversion}, but with different mode permutation. Similar relation can be established for $\mathscr{V}$ in terms of $\text{B}^{(3)}$ and $\text{B}^{(4)}$.

\subsection{Tensor Eigen Value Decomposition (EVD)}
As a generalization of matrix eigenvalues to tensors, several definitions exist in literature for tensor eigenvalues \cite{QiLEigen}. But most of these definitions applies to super-symmetric tensors which are defined as a class of tensors that are invariant under any permutation of their indices \cite{QiLBook}. Such an approach has applications in Physics and Mechanics \cite{QiLBook}. But there is no single generalization to tensor case that preserves all the properties of matrix eigenvalues  \cite{limsingular}. Here we present a particular generalization from \cite{LuEVD, MLTI2, TensorDet} which can be seen as the extension of matrix spectral decomposition theorem and has applications in multi-linear system theory. 

\begin{definition}
Let $\mathscr{A} \in \mathbb{C}^{I_1 \times \dots \times I_N \times I_1 \times \dots \times I_N }, \mathscr{X} \in \mathbb{C}^{I_1 \times \dots \times I_N }$, $\lambda \in \mathbb{C}$, where $\mathscr{X}$ and $\lambda$ satisfy $\mathscr{A}*_N \mathscr{X} =\lambda \mathscr{X}$, then we call $\mathscr{X}$ and $\lambda$ as \textit{eigentensor} and \textit{eigenvalue} of $\mathscr{A}$ respectively \cite{LuEVD}.
\end{definition} 

\begin{lemma}\label{RealEigenVal}
Eigenvalues $\lambda$ of a Hermitian tensor $\mathscr{A} \in \mathbb{C}^{I_1 \times \dots \times I_N \times I_1 \times \dots \times I_N }$ are real, i.e $\lambda \in \mathbb{R}$.
\end{lemma}
\begin{proof}
Consider ${\mathscr{X}} \in \mathbb{C}^{I_1 \times\dots \times I_N}$ and $\lambda \in \mathbb{C}$. Now let's assume ${\mathscr{A}}$ is hermitian. i.e. ${\mathscr{A}}^H = {\mathscr{A}}$
\begin{equation*}
\begin{aligned}
{\mathscr{A}}*_N {\mathscr{X}} &=\lambda {\mathscr{X}}\\
{\mathscr{X}} *_N {\mathscr{A}}^T  &=\lambda {\mathscr{X}} \,\,\,\,\,\,\,\,\text{(from (\ref{commutative}))}\\
({\mathscr{X}} *_N {\mathscr{A}}^T)^*  &=(\lambda {\mathscr{X}})^*\\
{\mathscr{X}}^* *_N {\mathscr{A}}^H  &=\lambda^* {\mathscr{X}}^*\\
{\mathscr{X}}^* *_N {\mathscr{A}} &=\lambda^* {\mathscr{X}}^* \,\,\,\,\,\,\,\,\text{(since } {\mathscr{A}}={\mathscr{A}}^H)\\
{\mathscr{X}}^* *_N {\mathscr{A}} *_N {\mathscr{X}} &=\lambda^* {\mathscr{X}}^* *_N {\mathscr{X}}\\
{\mathscr{X}}^* *_N (\lambda {\mathscr{X}}) &=\lambda^* {\mathscr{X}}^* *_N {\mathscr{X}} \,\,\,\,\,\,\,\,\text{(as } {\mathscr{A}}*_N {\mathscr{X}}=\lambda {\mathscr{X}})\\
\lambda {\mathscr{X}}^* *_N {\mathscr{X}} &=\lambda^* {\mathscr{X}}^* *_N {\mathscr{X}} \,\,\,\,\,\,\,\,\text{(as }\lambda \text{ is a scalar} )\\
\lambda &= \lambda^* \Rightarrow \lambda \in \mathbb{R} \,\,\,\,\,\,\,\,\text{(since } {\mathscr{X}} \text{ is non-zero})
\end{aligned}
\end{equation*}
\end{proof}

\begin{theorem}
The EVD of a Hermitian tensor $\mathscr{A} \in \mathbb{C}^{I_1 \times \dots \times I_N \times I_1 \times \dots \times I_N }$ is given as \textup{\cite{TamonTensorInversion}} :
\begin{equation}\label{EVDTensor}
\mathscr{A} = \mathscr{U}*_N \mathscr{D} *_N \mathscr{U}^H
\end{equation}
where $\mathscr{U} \in \mathbb{C}^{I_1 \times \dots \times I_N \times I_1 \times \dots \times I_N}$ is a unitary tensor and $\mathscr{D} \in \mathbb{C}^{I_1 \times \dots \times I_N \times I_1 \times \dots \times I_N}$ is a square pseudo-diagonal tensor, i.e. $\mathscr{D}_{i_1,\dots ,i_N, j_1, \dots ,j_N} = 0$ if $(i_1,\dots,i_N)\neq (j_1,\dots,j_N)$ with its non-zero values being the eigenvalues of $\mathscr{A}$ and $\mathscr{U}$ containing the eigentensors of $\mathscr{A}$.  
\end{theorem}
This theorem can be proven using Lemma \ref{TransformPropertyLemma}, details are provided in \cite{TamonTensorInversion,TensorDet}. The eigenvalues of $\mathscr{A}$ are same as the eigenvalues of $f_{I_1,\dots,I_N|I_1,\dots,I_N}(\mathscr{A})$ \cite{Bernstein}. We will refer to a tensor  $\mathscr{A} \in \mathbb{C}^{I_1\times \dots \times I_N \times I_1\times \dots \times I_N}$ as \textit{positive semi-definite}, denoted by $\mathscr{A} \succeq 0$ if all its eigenvalues are non-negative, which is same as $f_{I_1,\dots,I_N|I_1,\dots,I_N}(\mathscr{A})$ being a positive semi-definite matrix. A tensor is \textit{positive definite}, $\mathscr{A} \succ 0$, if all its eigenvalues are strictly greater than zero. A positive semi-definite pseudo-diagonal tensor $\mathscr{D}$, will have all its components non-negative. Its square root can be denoted as $\mathscr{D}^{1/2}$ which is also pseudo-diagonal positive semi-definite whose elements are the square root of elements of $\mathscr{D}$ such that $\mathscr{D}^{1/2}*_N\mathscr{D}^{1/2}=\mathscr{D}$. Similarly, if $\mathscr{D}$ is positive definite, its inverse can be denoted as $\mathscr{D}^{-1}$ which is also pseudo-diagonal whose non zero elements are the reciprocal of the corresponding elements of $\mathscr{D}$. Based on tensor EVD, we can also write the square root of any Hermitian positive semi-definite tensor as $\mathscr{A}^{1/2} = \mathscr{U}*_N \mathscr{D}^{1/2} *_N \mathscr{U}^H$ and inverse of any Hermitian positive definite tensor as $\mathscr{A}^{-1} = \mathscr{U}*_N \mathscr{D}^{-1} *_N \mathscr{U}^H$.  
It is straightforward to see that the singular values of a tensor $\mathscr{A} \in \mathbb{C}^{I_1\times \dots \times I_N \times J_1 \times \dots \times J_M}$ are the square root of the eigenvalues of tensor $\mathscr{A}^H *_N \mathscr{A}$. From SVD, if $\mathscr{A} = \mathscr{U}*_N \mathscr{D} *_M \mathscr{V}^H
$ , then 
\begin{equation}
\mathscr{A}^H *_N \mathscr{A} = (\mathscr{V}*_M \mathscr{D}^H *_N \mathscr{U}^H) *_N (\mathscr{U}*_N \mathscr{D} *_M \mathscr{V}^H) = \mathscr{V}*_M \mathscr{D}^H *_N \mathscr{D} *_M \mathscr{V}^H 
\end{equation}
where $\mathscr{D}^H *_N \mathscr{D}$ is the pseudo-diagonal tensor with eigenvalues of $\mathscr{A}^H *_N \mathscr{A}$ on its pseudo-diagonal which are square of the singular values obtained from the SVD of $\mathscr{A}$.

\begin{definition}{\textbf{Trace}}:
The trace of a tensor $\mathscr{A} \in \mathbb{C}^{I_1\times \dots \times I_N \times I_1\times \dots \times I_N}$ is defined as the sum of its pseudo-diagonal entries :
\begin{equation}
\tr(\mathscr{A})=\sum_{i_1,\dots,i_N}\mathscr{A}_{i_1,i_2,\dots,i_N,i_1,i_2,\dots,i_N}
\end{equation}
\end{definition}

\begin{definition}{\textbf{Determinant}}:
The determinant of a tensor $\mathscr{A} \in \mathbb{C}^{I_1\times \dots \times I_N \times I_1\times \dots \times I_N}$ is defined as the product of its eigenvalues i.e., if $\mathscr{A} = \mathscr{U}*_N \mathscr{D} *_N \mathscr{U}^H$, then
\begin{equation}
\det(\mathscr{A})=\prod_{i_1,\dots,i_N}\mathscr{D}_{i_1,i_2,\dots,i_N,i_1,i_2,\dots,i_N}
\end{equation}
The eigenvalues of $\mathscr{A}$ are the same as that of its matrix transformation, hence $\det(\mathscr{A})=\det(f_{I_1,\dots,I_N|I_1,\dots,I_N}(\mathscr{A}))$. Note that there exists other definitions in literature for determinant based on how one chooses to define the eigenvalues of tensors \cite{QiLDet}. The definition presented here is the same as the unfolding determinant in \cite{TensorDet}.  
\end{definition} 

\subsubsection*{Some properties of trace and determinant}
The following properties can be easily shown by writing the tensors component wise or using Lemma \ref{TransformPropertyLemma}.
\begin{enumerate}
\item For two tensors $\mathscr{A} \in \mathbb{C}^{I_1 \times \dots \times I_N}$ and $\mathscr{B} \in \mathbb{C}^{I_1 \times \dots \times I_N}$of same size and order $N$, 
\begin{equation}\label{Traceprop}
\mathscr{A}*_N \mathscr{B}= \mathscr{B}*_N \mathscr{A} = \tr(\mathscr{A}\circ \mathscr{B}) = \tr(\mathscr{B}\circ \mathscr{A})
\end{equation}
\item For tensors $\mathscr{A} \in \mathbb{C}^{I_1 \times \dots \times I_N \times J_1 \times \dots \times J_M }$ and $\mathscr{B} \in \mathbb{C}^{J_1 \times \dots \times J_M \times I_1 \times \dots \times I_N }$, we have :
\begin{equation}\label{traceprope}
\tr(\mathscr{A}*_M \mathscr{B}) = \tr(\mathscr{B}*_N \mathscr{A})
\end{equation}
\begin{equation}\label{Sylvester}
\det(\mathscr{I}_N + \mathscr{A}*_M \mathscr{B} ) = \det(\mathscr{I}_M + \mathscr{B}*_N \mathscr{A} ) 
\end{equation}
where $\mathscr{I}_N$ and $\mathscr{I}_M$ are identity tensors of order $2N$ and $2M$ respectively. To prove (\ref{Sylvester}), we can use lemma \ref{TransformPropertyLemma} and Sylvester's matrix determinant identity \cite{MatrixCookbook}.
\item For tensors $\mathscr{A},\mathscr{B} \in \mathbb{C}^{I_1 \times \dots \times I_N \times I_1 \times \dots \times I_N }$, we have :
\begin{equation}\label{eq18}
\det(\mathscr{A}*_N\mathscr{B})=\det(\mathscr{B}*_N\mathscr{A})=\det(\mathscr{A})\cdot \det(\mathscr{B})
\end{equation}
\item Trace of a positive semi-definite tensor is the sum of its eigenvalues.
\item The absolute value of the determinant of a unitary tensor is 1 and the determinant of a square pseudo-diagonal tensor is the product of its pseudo-diagonal entries.
\end{enumerate}

\subsection{Tensor LU decomposition}
LU decomposition is a powerful tool in Linear Algebra which can be used for solving system of equations. In order to solve systems of multi-linear equations, \cite{TensorDet} proposed an LU decomposition form for tensors. For $\mathscr{A} \in \mathbb{C}^{I_1 \times \dots \times I_N \times I_1 \times \dots \times I_N}$, the LU factorization takes the form :
\begin{equation}
\mathscr{A} = \mathscr{L} *_N \mathscr{U}
\end{equation}
where $\mathscr{L},\mathscr{U} \in \mathbb{C}^{I_1 \times \dots \times I_N \times I_1 \times \dots \times I_N}$ are pseudo-lower and pseudo-upper triangular tensors respectively. In order to solve a system of multi-linear equation $\mathscr{A}*_N \mathscr{X} = \mathscr{B}$ to find $\mathscr{X}$, LU decomposition of $\mathscr{A}$ can be used to break the equation into two pseudo-triangular equations  $\mathscr{U}*_N \mathscr{X} = \mathscr{Y}$ and $\mathscr{L}*_N \mathscr{Y} = \mathscr{B}$. These two equations can be solved using forward and backward substitution algorithms proposed in \cite{TensorDet}. When $\mathscr{B}$ is an identity tensor, this method can also be used for finding the inverse of a tensor. More details on computing LU decomposition and required conditions for its existence can be found in \cite{TensorDet}.

Note that all these tensor decompositions represent a given tensor in terms of a contracted product between factor tensors. Hence they all can be represented using Tensor Network diagrams. For example, we show the TN diagram corresponding to tensor SVD in Figure \ref{SVD_TN_Fig}. A detailed TN representation of several other tensor decompositions such as Tucker, PARAFAC, Tensor Train Decomposition, is also presented in \cite{cichocki2014era}.

\begin{figure}
\center
\includegraphics[scale=0.5]{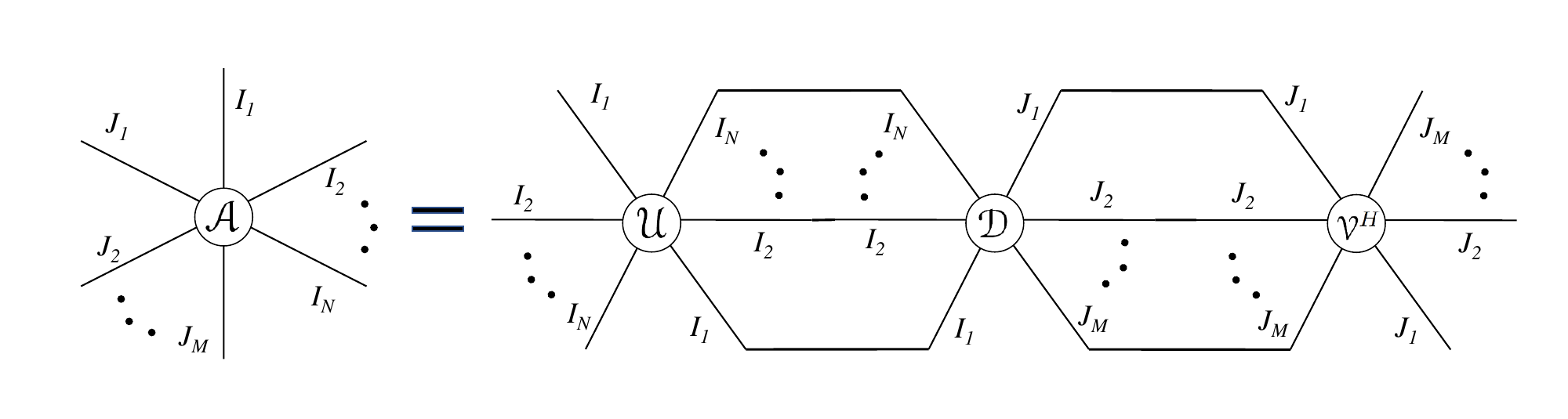}
\caption{TN representation of Tensor SVD from \eqref{SVDTensor} \label{SVD_TN_Fig}}
\end{figure}

\section{Multi-linear Tensor Systems}\label{Sec5}
Using the tools presented so far, we will now present the notions of multi-linear system theory using tensors.  
\subsection{Discrete time multi-linear tensor systems}
A discrete time multi-linear tensor system is characterized by an order $N+M$ system tensor $\mathscr{H}[k] \in \mathbb{C}_k^{J_1 \times \dots \times J_M \times I_1 \times \dots \times I_N}$ which produces an order $M$ output tensor sequence $\mathscr{Y}[k] \in \mathbb{C}^{J_1 \times \dots \times J_M}_k$ from an input tensor sequence  $\mathscr{X}[k] \in \mathbb{C}^{I_1 \times \dots \times I_N}_k$ through a discrete contracted convolution as defined in \eqref{dfeeq3}. The system tensor can be seen as an impulse response tensor whose $(j_1,\dots,j_M,i_1,\dots,i_N)$th entry is the impulse response from the $(i_1,\dots,i_N)$th input to the $(j_1,\dots,j_M)$th output. 

A system tensor is considered $p-$stable if corresponding to every input of finite $p-$norm, the system produces an output which is also finite $p-$norm. When $p \rightarrow \infty$, this notion is known as Bounded Input Bounded Output (BIBO) stability. The $\infty-$norm of a signal tensor $\mathscr{X}$ is essentially its peak amplitude evaluated over all the tensor components and all times, i.e.
\begin{equation}
\Vert \mathscr{X} \Vert_{\infty} = \sup_{k} \Vert \mathscr{X}[k] \Vert_{\infty} = \sup_{k} \max_{i_1,\dots,i_N} \mid \mathscr{X}_{i_{1},i_{2},....,i_{N}}[k]\mid. 
\end{equation}
  
\begin{theorem}\label{DTThBIBO}
For a discrete multi-linear time invariant system with order $N$ input and order $M$ output with an order $N+M$ impulse response system tensor $\mathscr{H}[k] \in \mathbb{C}^{J_1 \times \dots \times J_M \times I_1 \times \dots \times I_N}_k$, is BIBO stable if and only if 
\begin{equation}\label{BIBODT}
\max_{j_1,\dots,j_M} \sum_{i_1,\dots,i_N} \sum_k \mid \mathscr{H}_{j_1,\dots,j_M,i_1,\dots,i_N}[k] \mid < \infty. 
\end{equation} 
\begin{proof}
If the input signal tensor $\mathscr{X}[k]$ satisfies $\Vert \mathscr{X} \Vert_{\infty} < \infty$, then output is given via discrete contracted convolution as:
\begin{equation}
\mathscr{Y}[k] = \sum_{l} \mathscr{H}[k-l] *_N \mathscr{X}[l] 
\end{equation}
and thus,
\begin{align}
\max_{j_1,\dots,j_M} \mid \mathscr{Y}_{j_1,\dots,j_M}[k] \mid &= \max_{j_1,\dots,j_M} \mid \sum_{l} \sum_{i_1,\dots,i_N} \mathscr{H}_{j_1,\dots,j_M,i_1,\dots,i_N}[k-l]  \mathscr{X}_{i_1,\dots,i_N}[l]  \mid 
\\
&\leq \Big( \max_{j_1,\dots,j_M} \sum_{l} \sum_{i_1,\dots,i_N} \mid \mathscr{H}_{j_1,\dots,j_M,i_1,\dots,i_N}[k-l] \mid \Big) \max_{i_1,\dots,i_N} \sup_l \mid \mathscr{X}_{i_1,\dots,i_N}[l] \mid
\\
&= \Big( \max_{j_1,\dots,j_M} \sum_{l} \sum_{i_1,\dots,i_N} \mid \mathscr{H}_{j_1,\dots,j_M,i_1,\dots,i_N}[k-l] \mid \Big) \Vert \mathscr{X} \Vert_{\infty}.
\end{align}
Hence we get,
\begin{align}
\Vert \mathscr{Y} \Vert_{\infty} &= \sup_{k} \Vert \mathscr{Y}[k] \Vert_{\infty} = \sup_{k} \max_{j_1,\dots,j_M} \mid \mathscr{Y}_{j_{1},....,j_{M}}[k]\mid 
\\
& \leq \Big( \max_{j_1,\dots,j_M} \sum_{i_1,\dots,i_N} \sum_k \mid \mathscr{H}_{j_1,\dots,j_M,i_1,\dots,i_N}[k] \mid \Big) \Vert \mathscr{X} \Vert_{\infty} < \infty.  
\end{align}
which proves that output is bounded if \eqref{BIBODT} is satisfied.
To prove of the converse of the theorem, it suffices to show any example where if \eqref{BIBODT} is not satisfied, there exists a bounded input which leads to an unbounded output. For this, we can simply consider the case where input and output are scalars which is a special case of the tensor formulation. Equation \eqref{BIBODT} in that case translates to BIBO condition for SISO LTI system, i.e. $\sum_k |h[k]| < \infty $. Hence it can be readily verified that a signum input defined as $x[n]=sgn(h[-n])$ which is bounded will lead to an unbounded output if the impulse response sequence is not absolutely summable \cite{mitra2006digital}.

\end{proof}
\end{theorem}

The BIBO stability condition for a MIMO LTI system requires that every element of the impulse response matrix must be absolutely summable. The condition from \eqref{BIBODT} can be seen as an extension to the tensor case, where every element of the impulse response tensor must be absolutely summable.
Furthermore, we extend the definitions of poles and zeros from matrix-based systems to tensors. A matrix transfer function $\breve{\text{H}}(z)$ has a pole at frequency $\nu$ if some entry of $\breve{\text{H}}(z)$ has a pole at $z=\nu$ \cite{PolesandZerosMIMO}. Also, $\breve{\text{H}}(z)$ has a zero at frequency $\gamma$ if the rank of $\breve{\text{H}}(z)$ drops at $z=\gamma$ \cite{PolesandZerosMIMO}. Similarly, a tensor transfer function $\breve{\mathscr{H}}(z)$ has a pole  at frequency $\nu$ if some entry of $\breve{\mathscr{H}}(z)$ has a pole at $z=\nu$. Also, $\breve{\mathscr{H}}(z)$ has a zero at frequency $\gamma$ if the rank of $f_{J_1,\dots,J_M|I_1,\dots,I_N}(\breve{\mathscr{H}}(z))$ drops at $z=\gamma$. Such a rank is also sometimes referred to as the unfolding rank of the tensor \cite{TensorDet, MLTI2}. 

A tensor system is BIBO stable if all its components are BIBO stable. This implies that every pole of every entry of its transfer function has a magnitude less than 1, i.e. all the poles lie within the unit circle on the z-plane. 
  
\subsection{Continuous time multi-linear tensor systems}

A continuous time multi-linear tensor system is characterized by an order $N+M$ system tensor $\mathscr{H}(t) \in \mathbb{C}_t^{J_1 \times \dots \times J_M \times I_1 \times \dots \times I_N}$ which produces an order $M$ output tensor signal $\mathscr{Y}(t) \in \mathbb{C}^{J_1 \times \dots \times J_M}_t$ from an input tensor signal  $\mathscr{X}(t) \in \mathbb{C}^{I_1 \times \dots \times I_N}_t$ through a contracted convolution as defined in \eqref{cfeeq2}. The system tensor can be seen as an impulse response tensor whose $(j_1,\dots,j_M,i_1,\dots,i_N)$th entry is the impulse response from the $(i_1,\dots,i_N)$th input to the $(j_1,\dots,j_M)$th output. 

Similar to the discrete case, a continuous system tensor is considered $p-$stable if corresponding to every input of finite $p-$norm, the system produces an output which is also finite $p-$norm. This notion is known as Bounded Input Bounded Output (BIBO) stability if $p=\infty$. The $\infty-$norm of a continuous signal tensor $\mathscr{X}$ is essentially its peak amplitude evaluated over all the tensor components and all times, i.e.
\begin{equation}
\Vert \mathscr{X} \Vert_{\infty} = \sup_{t} \Vert \mathscr{X}(t) \Vert_{\infty} = \sup_{t} \max_{i_1,\dots,i_N} \mid \mathscr{X}_{i_{1},i_{2},....,i_{N}}(t)\mid. 
\end{equation}
  
\begin{theorem}\label{CTThBIBO}
For a continuous multi-linear time invariant system with order $N$ input and order $M$ output with an order $M+N$ impulse response system tensor $\mathscr{H}(t) \in \mathbb{C}^{J_1 \times \dots \times J_M \times I_1 \times \dots \times I_N}_t$, is BIBO stable if and only if 
\begin{equation}\label{BIBODTC}
\max_{j_1,\dots,j_M} \sum_{i_1,\dots,i_N} \int \mid \mathscr{H}_{j_1,\dots,j_M,i_1,\dots,i_N}(t) \mid dt < \infty. 
\end{equation} 
\end{theorem}
The condition from \eqref{BIBODTC} implies that every element of the impulse response tensor must be absolutely integrable. The proof of Theorem \ref{CTThBIBO} follows the same line of proof as of Theorem \ref{DTThBIBO}. Furthermore, a continuous system tensor with transfer function $\bar{\mathscr{H}}(\omega)$ is BIBO stable if all its components are BIBO stable. This implies that every pole of every entry of its transfer function has a real part less than 0.

\subsection{Applications of multi-linear tensor systems}

A tensor multi-linear (TML) system can be used to model and represent processes where two tensor signals are coupled  through a multi-linear functional. Among various other applications, use of tensors is ubiquitous in modern communication systems where the signals and systems involved have an inherent multi-domain structure. The physical layer model of modern communication systems invariably spans more than one domain of transmission and reception such as space, time, frequency to name a few. Consequently, the associated signal processing at the transmitter and receiver has to be cognizant of the multiple domains and their mutual effect on each other for an efficient resource utilization. Hence the signals and the systems involved are best represented using tensors.

\subsubsection{TN representation of TML systems}
Consider a multi-domain communication system where the input signal is an order $N$ tensor $\mathscr{X}(t) \in \mathbb{C}_t^{I_1 \times \dots \times I_N}$, which passes through an order $M+N$ multi-linear channel, $\mathscr{H}(t) \in \mathbb{C}_t^{J_1 \times \dots \times J_M \times I_1 \times \dots \times I_N}$ to generate an output of order $M$ using contracted convolution as $\mathscr{Y}(t) = \mathscr{H}(t) \bullet_N \mathscr{X}(t)$.  The channel is expressed as a TML system and its coupling with the input in both frequency and time domain can be represented in a TN diagram as shown in Figure \ref{ContConv_Fig3}. Note that each edge of the input is connected with the common edge of the channel via a dotted line in time domain and via a solid line in frequency domain. Instead of a regular block diagram representation of such systems, the TN diagram has the advantage that it graphically details all the modes of the input and the channel. Thus, just by looking at the free edges of the overall TN diagram, one can determine the modes of the output. The linearity of the system is reflected in the fact that any given edge of the channel is connected with a single edge of the input. Thus a TML system is easy to identify visually in a TN diagram by observing the presence of one on one edge connections between the system and the input.

Note that in a communication system, the input signal is often precoded before transmission by a transmit filter and the output signal is processes via a receive filter. The transmit and receiver filter can also be considered as system tensors. Thus the TML channel $\mathscr{H}(t)$ can be seen as a cascade of three system tensors. Let the transmit filter be represented by $\mathscr{H}_T(t) \in \mathbb{C}_t^{K_1 \times \dots \times K_Q \times I_1 \times \dots \times I_N}$ which transforms the order $N$ input into an order $Q$ transmit signal. The physical channel between the source and destination is modelled as an order $P+Q$ tensor $\mathscr{H}_C(t) \in \mathbb{C}_t^{L_1 \times \dots \times L_P \times K_1 \times \dots \times K_Q }$, and the receive filter is represented by  $\mathscr{H}_R(t) \in \mathbb{C}_t^{J_1 \times \dots \times J_M \times L_1 \times \dots \times L_P}$. In this case, the equivalent channel $\mathscr{H}(t)$ is obtained via a cascade of the three system tensors as :
\begin{equation}
\mathscr{H}(t) = \mathscr{H}_R(t) \bullet_P \mathscr{H}_C(t) \bullet_Q \mathscr{H}_T(t).
\end{equation}
A detailed derivation of such a channel representation can be found in \cite{AdithyaPaper}. A cascade of TML systems is conveniently represented in a TN diagram as shown in Figure \ref{ContConv_Fig4} which illustrates the coupling of the receive filter, physical channel and transmit filter system tensors in both time and frequency domains. Hence the nodes for $\mathscr{H}(t)$ and $\bar{\mathscr{H}}(\omega)$ in Figure \ref{ContConv_Fig3} can be broken down into component system tensors from Figure \ref{ContConv_Fig4}. A tensor system has multiple modes, and the contraction can be along various combinations of such modes. Hence the TN representation becomes extremely useful as opposed to regular block diagrams, since it allows to depict the state and coupling of each mode.  
 
\begin{figure}
\center
\includegraphics[scale=0.5]{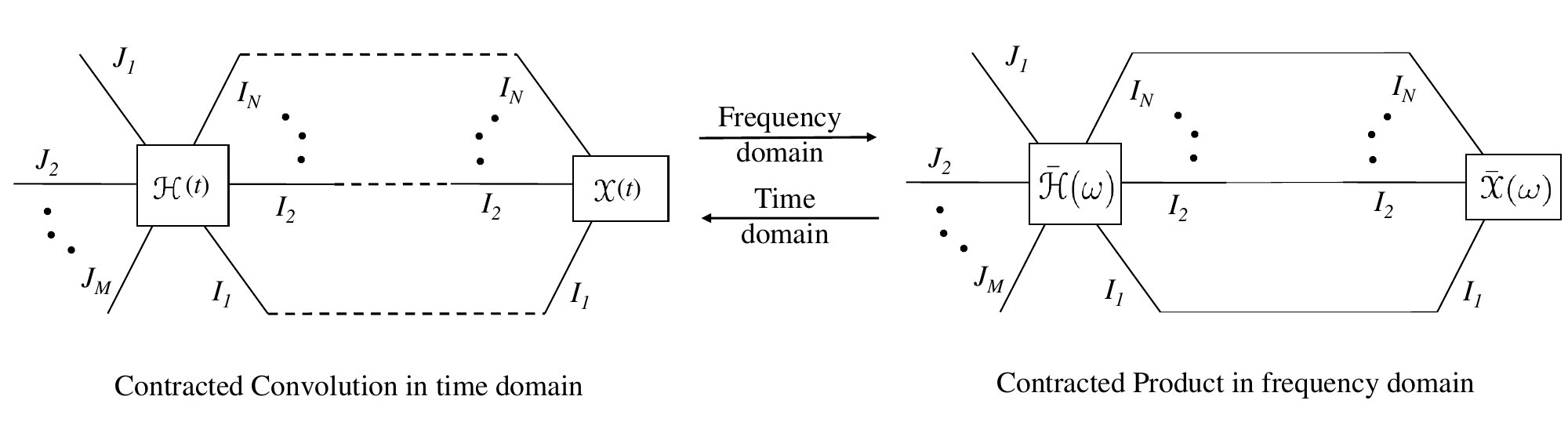}
\caption{TN representation of TML system in time and frequency domain. \label{ContConv_Fig3}}
\end{figure}

\begin{figure}
\center
\includegraphics[scale=0.55]{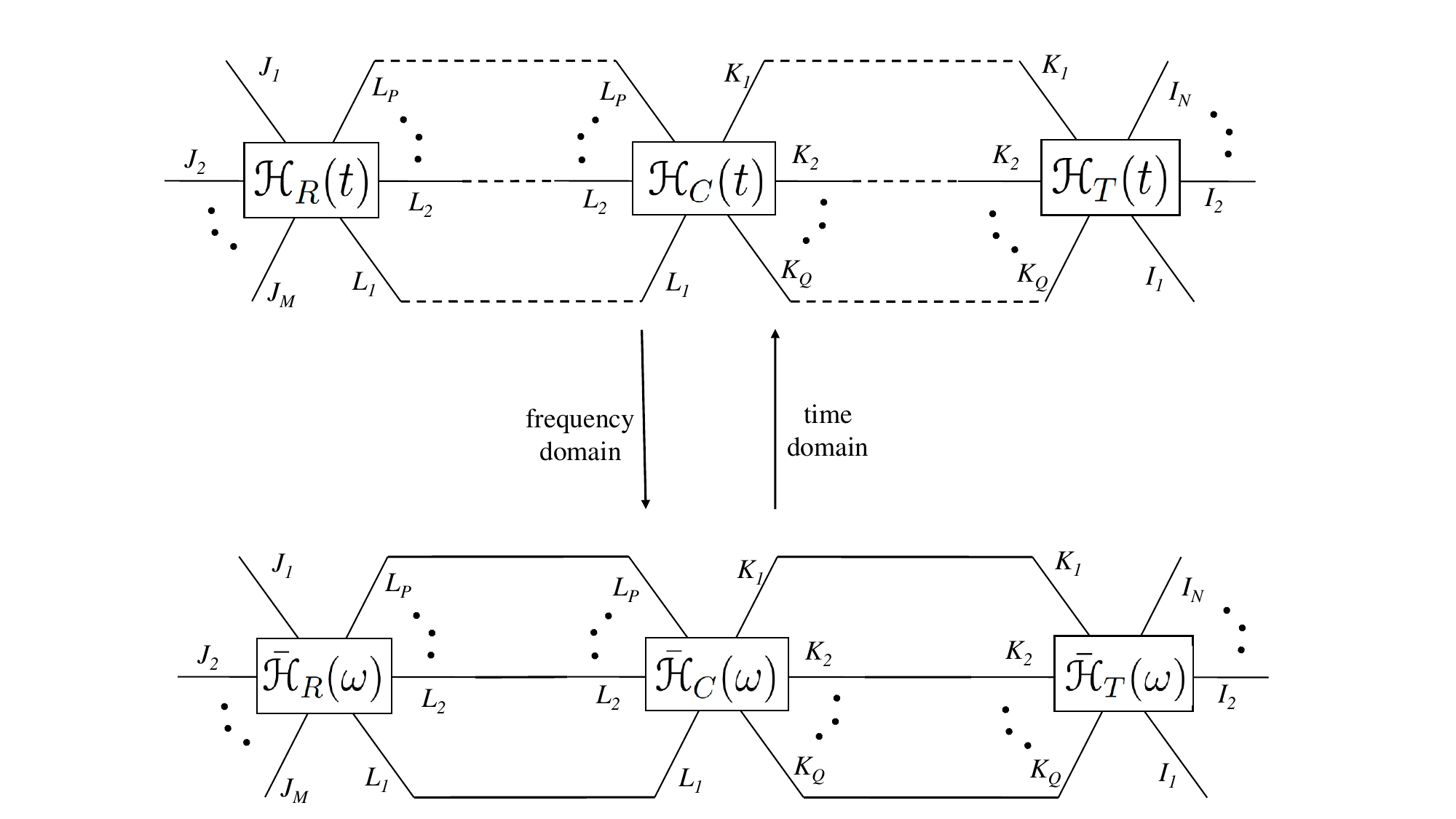}
\caption{TN representation of equivalent multi linear channel in time and frequency domain. \label{ContConv_Fig4}}
\end{figure}

\subsubsection{Example of Tensor Contraction for MIMO CDMA systems}

Code Division Multiple Access (CDMA) is a spread spectrum technique used in communication systems where multiple transmitting users can send information simultaneously over a single communication channel, thereby enabling multiple access. Each user uses the entire bandwidth along with a distinct pseudo-random spreading code to transmit information which is used to distinguish the users at the receiver. More details on CDMA can be found in \cite{proakis2007fundamentals}.

Consider an uplink scenario where $K$ users are transmitting information to a single base station (BS). Assume a simple additive white Gaussian Noise (AWGN) channel. Each user is assigned a distinct spreading sequence denoted by vector $\underline{\text{s}}^{(k)} \in \mathbb{C}^{L}$ of length $L$ which transmits a symbol $x^{(k)}$ for user $k$. The received signal at the BS can be written as \cite{MIMOCDMA2} :
\begin{equation}
\underline{\text{y}} = \sum_{k=1}^{K} x^{(k)} \underline{\text{s}}^{(k)} + \underline{\text{z}} 
\end{equation} 
where $\underline{\text{z}} \in \mathbb{C}^L$ represents the noise vector.
Now consider the extension of such a system model in the presence of flat fading channel and multiple antennas. Assume $K$ users each with $N_T$ transmit antennas are transmitting simultaneously to a BS with $N_R$ receive antennas. To allow multiple access, all the transmit antennas of all different users are assigned different spreading sequences of length $L$. Let $\underline{\text{s}}^{(k,i)} \in \mathbb{C}^L$ denotes the length $L$ spreading vector for the data transmitted by the $i$th antenna of the $k$th user, $x^{(k,i)}$. Transmit symbols are assumed to have zero mean and energy $E_s = \mathbb{E}[|x^{(k,i)}|^2]$, and the transmit vector from each user and each antenna is generated as $x^{(k,i)}\underline{\text{s}}^{(k,i)}$. The MIMO communication channel between user $k$ and the BS is defined as a matrix $\text{H}^{(k)} \in \mathbb{C}^{N_R \times N_T}$ where the random channel matrix has independent and identically distributed (i.i.d.) zero mean circular symmetric complex Gaussian entries with variance $1/N_R$. The distribution is denoted as $\mathcal{CN}(0,1/N_R)$. The received signal $\text{Y} \in \mathbb{C}^{N_R \times L}$ can be written as \cite{MIMOCDMA2, MIMOCDMA3}:
\begin{equation}\label{HXSplusZbefore}
\text{Y} = \sum_{k=1}^K \text{H}^{(k)} \text{X}^{(k)}\text{S}^{(k)} + \text{Z},
\end{equation}  
where $\text{X}^{(k)}$ is an $N_T \times N_T$ diagonal matrix defined as $diag(x^{(k,1)},x^{(k,2)},\dots,x^{(k,N_T)})$ and $\text{S}^{(k)}$ is an $N_T \times L$ matrix defined as $(\underline{\text{s}}^{(k,1)T}, \underline{\text{s}}^{(k,2)T}, \dots, \underline{\text{s}}^{(k,N_T)T})^T$. Also $\text{Z}$ represents $N_R \times L$ noise matrix with i.i.d. components distributed as $\mathcal{CN}(0,N_0)$. In \cite{MIMOCDMA2}, a per user matched filter receiver is considered for such a system by assuming the interference from other users as noise. It is shown in \cite{MIMOCDMA2} that such a receiver under performs as compared to a multi-user receiver which detects the transmit symbols for all the users together. Hence several multi-user receivers are presented in \cite{MIMOCDMA2}  by rewriting the system model from \eqref{HXSplusZbefore} as :
\begin{equation}\label{HXSplusZ}
\text{Y} = \text{H}\bar{\text{X}}\text{S} + \text{Z}
\end{equation}
where $\text{H}=(\text{H}^{(1)},\dots,\text{H}^{(K)}) \in \mathbb{C}^{N_R \times K \cdot N_T}$, $\bar{\text{X}}=diag(x^{(1,1)},\dots, x^{(1,N_T)},\dots, x^{(K,1)},\dots,x^{(K,N_T)}) \in \mathbb{C}^{K\cdot N_T \times K \cdot N_T}$, and $\text{S}=(\text{S}^{(1)T},\dots, \text{S}^{(K)T})^T \in \mathbb{C}^{K \cdot N_T \times L}$. Based on this, a multi-user receiver which aims to mitigate the effects of $\text{H}$ (spatial interference) and $\text{S}$ (multiple access interference) is considered. The received signal is linearly processed in two stages as $\text{Y} \rightarrow \text{AY} \rightarrow \text{AYB}$ and the transmit signal is decoded as \cite{MIMOCDMA2} : 
\begin{equation} \label{LMMSE12}
\hat{x}^{(k,i)} = \underset{x \in \mathcal{D}}{\text{arg max}} |(\text{AYB})_{j,j} - x|^2, \quad \text{where }j=(k-1)N_T+i.  
\end{equation}
where $\mathcal{D}$ denotes the set of symbols in the transmit constellation map. Essentially $\text{AYB}$ represents an estimated version of matrix $\bar{\text{X}}$, whose diagonal elements at index $j$ are used to decode the transmitted symbols and map them back to index $(k,i)$. The matrices $\text{A}$ and $\text{B}$ separately aims to mitigate the effects of spatial interference and multiple access interference on the received signal, and are defined as : 
\begin{equation}\label{AeqLMMSE}
  \text{A} \triangleq \begin{cases}
    (\text{H}^H\text{H})^{-1}\text{H}^H, & \text{ZF}.\\
    (\text{H}^H\text{H} + \dfrac{N_0}{E_s} \text{I}_{KN_T})^{-1}\text{H}^H, & \text{LMMSE}.
  \end{cases}
\end{equation}
and
\begin{equation}\label{BeqLMMSE}
  \text{B} \triangleq \begin{cases}
    \text{S}^H(\text{S}\text{S}^H)^{-1}, & \text{DECOR}.\\
    \text{S}^H(\text{S}\text{S}^H + \dfrac{N_0}{E_s} \text{I}_{KN_T})^{-1}, & \text{LMMSE}.
  \end{cases}
\end{equation}
where $\text{I}_{KN_T}$ is an identity matrix of size $K\cdot N_T \times K\cdot N_T$. The zero forcing (ZF) receiver ignores the impact of noise and only tries to counter the effect of channel, while the linear minimum mean square error (LMMSE) receiver tries to reduce the noise while simultaneously aiming to mitigate the effect of channel. The DECOR choice represents a multi user decorrelator receiver. At high SNR value, i.e. as $N_0/E_s \rightarrow 0$, the LMMSE options reduces to ZF and DECOR.  

Such a receiver based on jointly processing all the users gives better performance than a per user receiver \cite{MIMOCDMA2}. But it still has a drawback that it tries to combat the spatial interference and multiple access interference separately in two stages. Moreover, while the input in \eqref{HXSplusZ} is represented as a matrix $\bar{\text{X}}$, only its diagonal elements contain the transmit elements which is formed from the concatenation of various $x^{(k,i)}$. Thus such a system model does not fully exploit the multi-linearity of the system and tries to force a linear structure by manipulating the entities involved in order to fit the vector based well known LMMSE, ZF or DECOR solutions. In fact, the tensor framework can be perfectly used to represent such a system model while keeping the natural structure of the system intact, and develop a tensor multi-linear (TML) receiver.

Since the input symbol $x^{(k,i)}$ is indexed by two indices $k$ and $i$, it is natural to represent the input as a matrix $\text{X}$ of size $K \times N_T$ with elements $\text{X}_{k,i}=x^{(k,i)}$. Further, the input signal is transmitted as a vector $\underline{\text{x}}^{(k,i)}$ of length $L$ corresponding to each user index $k$ and antenna index $i$. Hence the transmitted signal through the channel can be represented as a third order tensor $\mathscr{X} $ of size $K \times N_T \times L$ where $\mathscr{X}_{k,i,l} = \underline{\text{x}}^{(k,i)}_l$. To generate $\mathscr{X}$ from $\text{X}$, we define the spreading sequences as an order 5 tensor  $\mathscr{S} \in \mathbb{C}^{K \times N_T \times L \times K \times N_T}$ with elements $\mathscr{S}_{k,i,l,k',i'} = \underline{\text{s}}_{l}^{(k,i)}$ when $k=k',i=i',$ for all $l$, and 0 elsewhere. Then we have $\mathscr{X} = \mathscr{S}*_2 \text{X}$. Note that here we assume that the elements of $\text{X}$ are mapped one to one with a spreading sequence hence the entries of $\mathscr{S}$ corresponding to $k \neq k', i \neq i'$ are zero. In certain applications, a linear combination of the input symbols might be transmitted, in which case the structure of $\mathscr{S}$ which represents a transmit filtering operation, will change accordingly. The channel matrices $\text{H}^{(k)}$ corresponding to each user can be represented as a slice in a third order tensor $\mathscr{H} \in \mathbb{C}^{N_R \times K \times N_T}$ where $\mathscr{H}_{:,k,:} = \text{H}^{(k)}$. Thus the system model can be given as :
\begin{equation}\label{TensorYHSXmodel}
\text{Y} = \underbrace{\mathscr{H} *_2 \mathscr{S}}_{\bar{\mathscr{H}}} *_2 \text{X} + \text{Z}
\end{equation} 
where $\bar{\mathscr{H}} \in \mathbb{C}^{N_R \times L \times K \times N_T}$ represents the equivalent fourth order TML channel between the order two input $\text{X}$ and order two output $\text{Y}$.

Note the advantage in modelling the system model through \eqref{TensorYHSXmodel} is that all the associated entities retain their natural structure, and  a joint TML receiver can be designed to combat the effect of all the interferences of all the users simultaneously. A multi-linear minimum mean square error receiver which acts across all the domains simultaneously can be represented through a tensor $\mathscr{R} \in \mathbb{C}^{ K \times N_T \times N_R \times L}$ which produces an estimate of the input $\text{X}$ by acting upon the received tensor $\text{Y}$ as $\tilde{\text{X}} = \mathscr{R} *_2 \text{Y}$. Thus each element of the estimated input at the receiver is a linear combination of all the elements of $\text{Y}$, where the coefficients of the linear combinations are encapsulated in $\mathscr{R}$. An optimal choice of $\mathscr{R}$ which minimizes the mean square error between $\text{X}$ and $\tilde{\text{X}}$, defined as $\mathbb{E}[||\text{X} - \tilde{\text{X}}||^2_2]$, is given as \cite{MMSEJournal} :
\begin{equation}
\mathscr{R} = \bar{\mathscr{H}}^H *_2 (\bar{\mathscr{H}} *_2 \bar{\mathscr{H}}^H + \dfrac{N_0}{E_s} \mathscr{I})^{-1}
\end{equation}
where $\mathscr{I}$ is an identity tensor of size $K \times N_T \times K \times N_T$. The estimated symbol $\tilde{\text{X}}$ can be used to detect the transmit symbols as:
\begin{equation}
\hat{x}^{(k,i)} = \underset{x \in \mathcal{D}}{\text{arg max}} |\tilde{\text{X}}_{k,i} - x|^2
\end{equation} 
We will refer to such a receiver as a TML MMSE receiver. Since a TML MMSE receiver jointly acts upon symbols across all domains, it aids in detecting the transmit symbol by  exploiting the multi-domain interference terms. Through simulation results, we will compare the performance of TML MMSE receiver with \eqref{LMMSE12}. In \eqref{LMMSE12}, we assume $\text{A}$ to be the LMMSE matrix from \eqref{AeqLMMSE}, and for $\text{B}$, we simulate both the DECOR and LMMSE matrices from \eqref{BeqLMMSE}. Hence we simulate LMMSE-DECOR and LMMSE-LMMSE cases from \cite{MIMOCDMA2}. We will refer to the former as LMMSE1 and the latter as LMMSE2 for our discussion going forward. The simulation parameters used are the same as in \cite{MIMOCDMA2} where entries of $\text{H}^{(k)}$ are i.i.d. which are $\mathcal{CN}(0,1/N_R)$. It is assumed that the channel realizations are known at the receiver. We assume uncoded transmission with QAM modulation where symbols are normalized to have unit energy, i.e. $E_s=1$. The spreading sequences are generated with i.i.d. symbols equiprobable over the set $\{\pm L^{-1/2}, \pm j L^{-1/2}\}$. We user Bit Error Rate (BER) and normalized mean square error (NMSE) as performance measures. All the results are plotted against  $E_b/N_0$ in dB where $E_b$ is the energy per bit defined as $E_s/2$. Thus $E_b/N_0$ represents received SNR per bit. We perform Monte Carlo simulations where the results are averaged over 100 different channel realizations, and at least 100 bit errors were collected for each SNR to calculate BER. The mean square error is normalized with respect to the number of elements in $\text{X}$ and is thus defined as $\text{NMSE} =||\text{X}-\tilde{\text{X}}||_F^2/(N_T \cdot K) $. To compare the performance difference between TML MMSE, LMMSE1, and LMMSE2 against SNR, we plot the BER and NMSE for the three receivers in Figures \ref{Fig_CDMA_3} and \ref{Fig_CDMA_4} respectively. We take $L=32, K=4,N_T=4, N_R = 32$. It can be clearly seen in both the figures, that the BER and the NMSE decreases as SNR increases. In particular the TML MMSE leads to a much lower BER and NMSE compared to the other two receivers as it exploits the multi-linearity of the equivalent channel to jointly combat interference across all the domains. Within LMMSE1 and LMMSE2, it can be observed that LMMSE2 performs better as the choice of DECOR for $\text{B}$ from \eqref{BeqLMMSE} is sub optimal as compared to LMMSE.

\begin{figure}
\center
\includegraphics[scale=0.5]{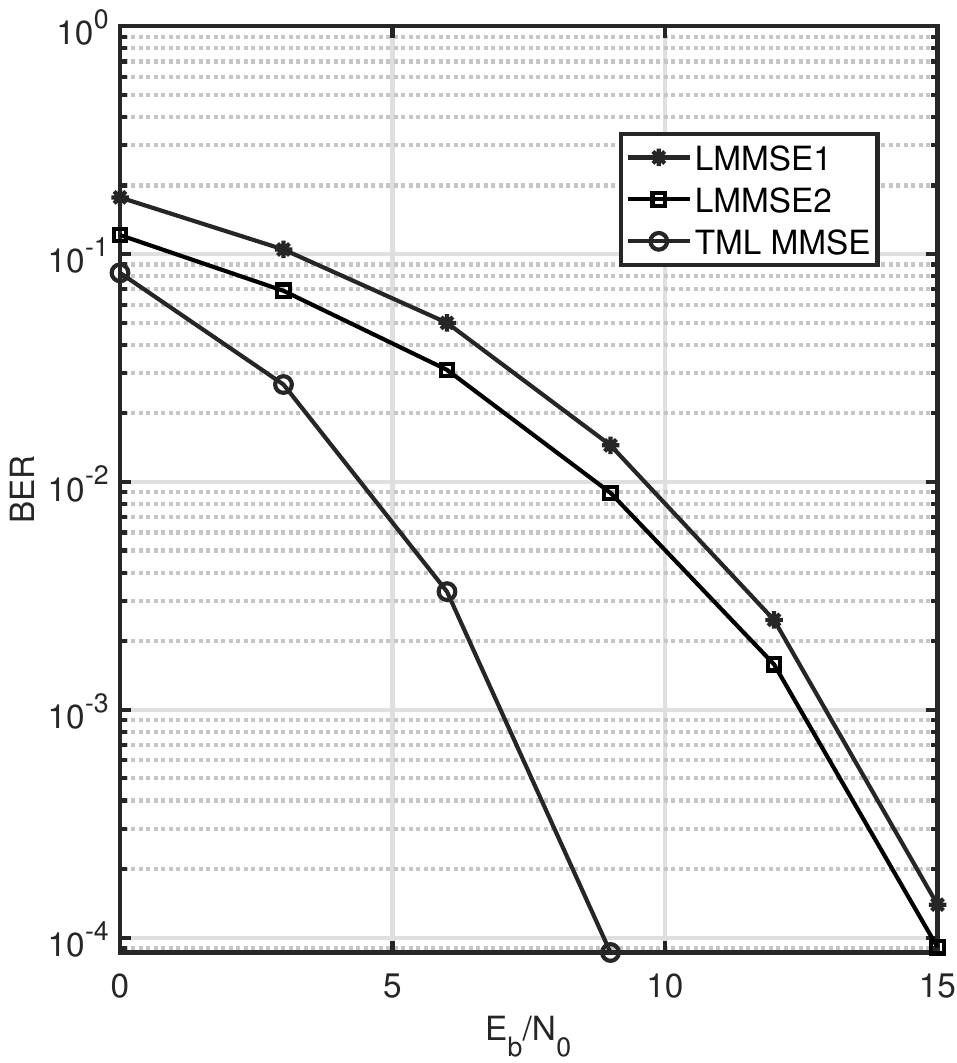}
\caption{BER performance for different receivers against SNR for $L=32, K=4, N_R =32, N_T =4$. \label{Fig_CDMA_3}}
\end{figure}

\begin{figure}
\center
\includegraphics[scale=0.5]{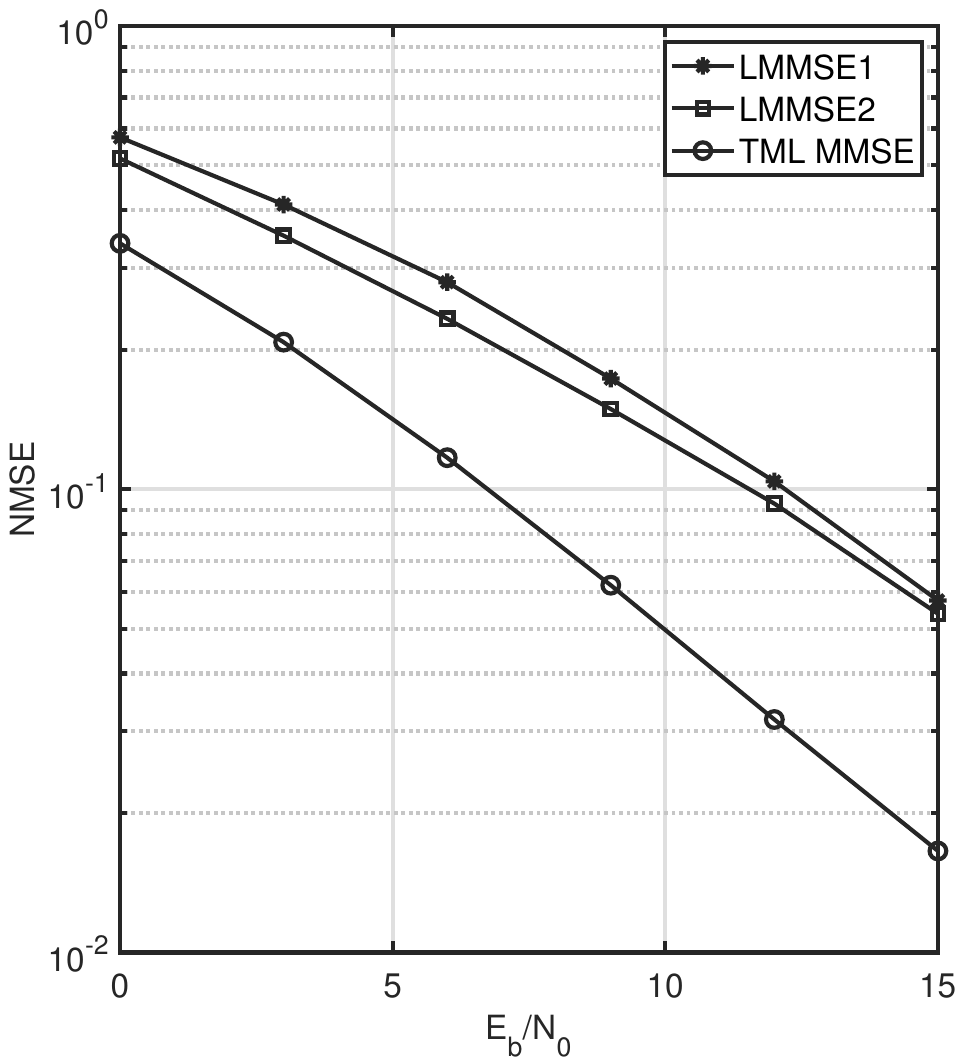}
\caption{Normalized MSE for different receivers against SNR for $L=32, K=4, N_T =32, N_T =4$. \label{Fig_CDMA_4}}
\end{figure}

Further, the advantage of TML MMSE can be clearly seen when BER and NMSE are observed for a fixed SNR per bit and variable number of users. Consider $L=64, N_R = 64, N_T=2$ and number of users $K$ is variable. Figures \ref{Fig_CDMA_5} and \ref{Fig_CDMA_6} present BER and NMSE performance against $K$ for two fixed values of SNR per bit. The solid lines correspond to a 5dB SNR per bit and dashed lines correspond to an 8dB SNR per bit. It can be clearly seen that for a fixed SNR per bit, the BER and NMSE curves for TML MMSE case remain almost flat as the number of users increase. On the other hand, the performance of LMMSE1 and LMMSE2 significantly degrades with increase in the number of users. As the number of users increase, the interference across domains also increases which is only efficiently utilized in the TML MMSE receiver. Hence it shows a robustness in its performance with increasing number of users.

\begin{figure}
\center
\includegraphics[scale=0.5]{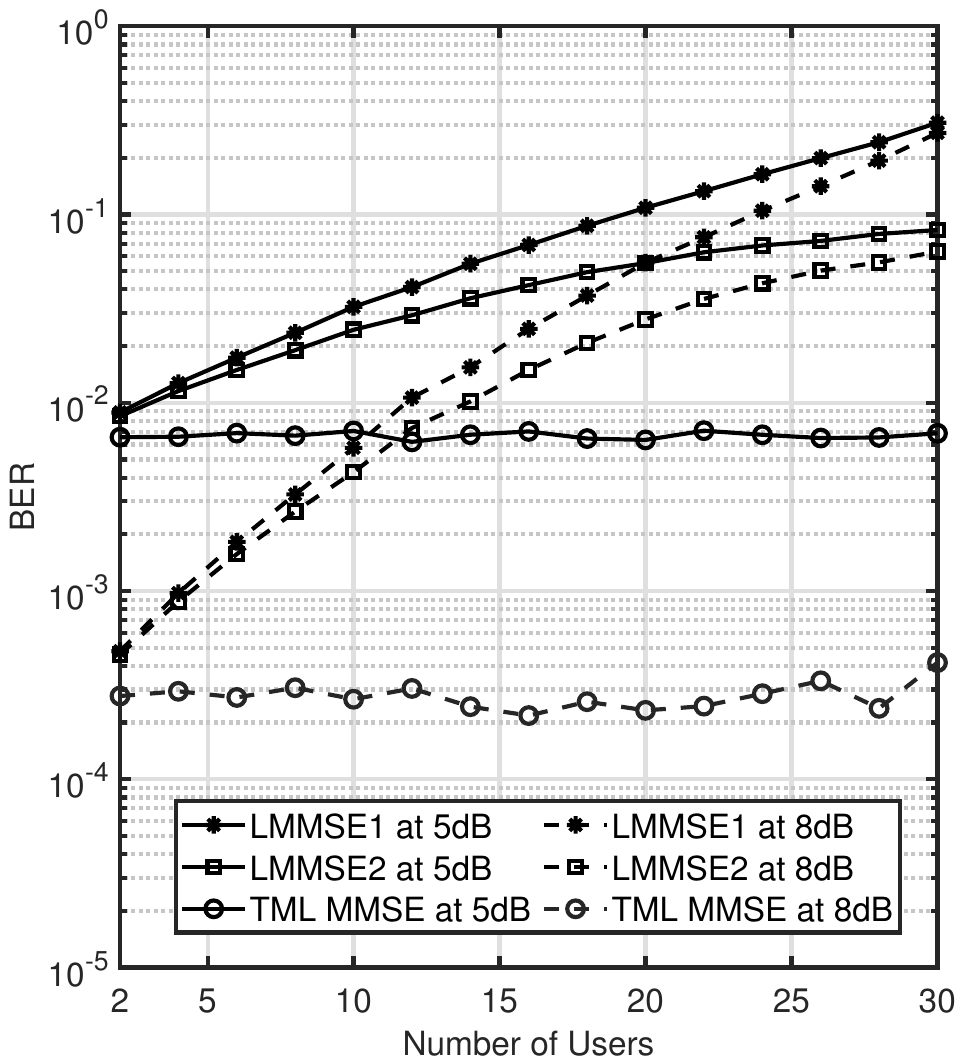}
\caption{BER performance for different receivers against number of users. \label{Fig_CDMA_5}}
\end{figure}

\begin{figure}
\center
\includegraphics[scale=0.5]{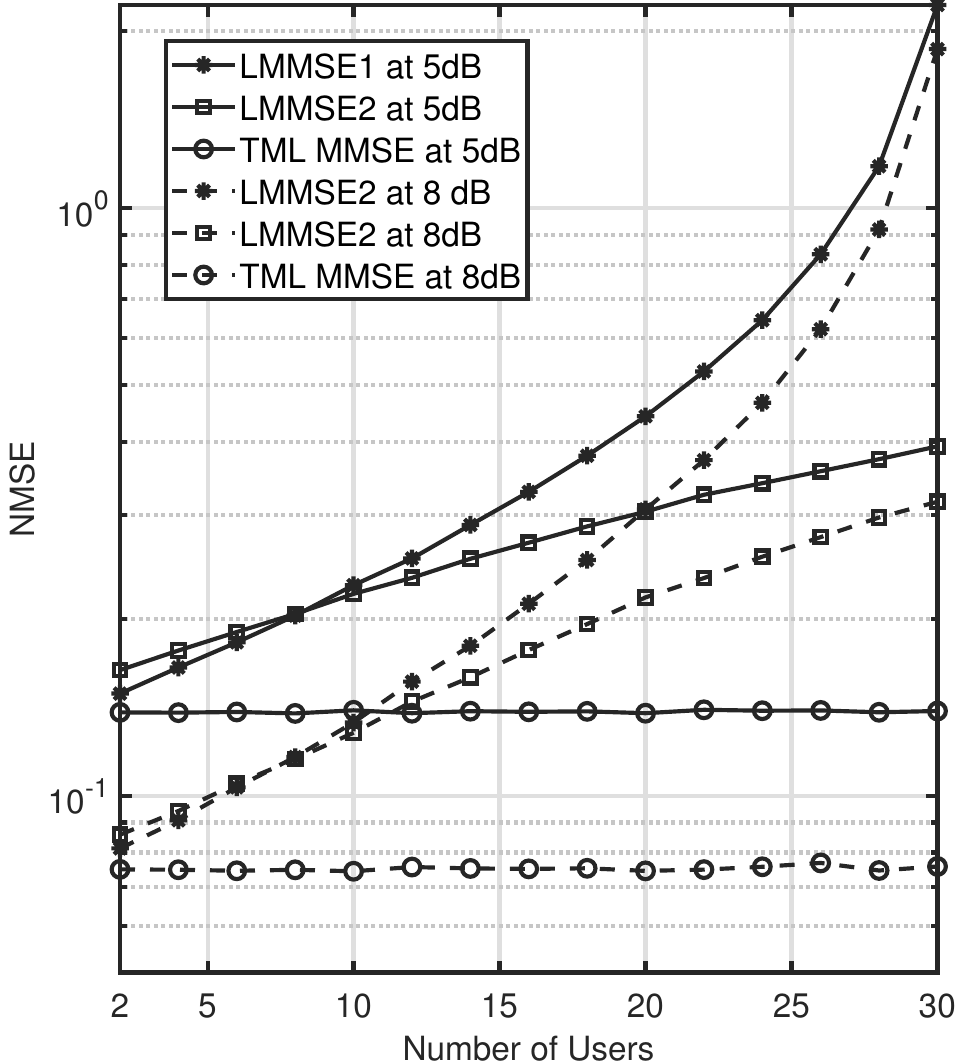}
\caption{Normalized MSE for different receivers against number of users. \label{Fig_CDMA_6}}
\end{figure}

Note that \eqref{TensorYHSXmodel} can be re-written as a system of linear equations by using vectorization of the input, output and noise, and considering the channel as a concatenated matrix $f_{N_R,L|K,N_T}(\bar{\mathscr{H}})$. Subsequently a joint receiver can also be designed using the transformed matrix channel, as presented in \cite{MIMOCDMA2}, which is conceptually equivalent to the TML MMSE approach presented in this paper. However, the concatenation of various domains obscures the different domain representations (indices) in the system. Such an approach makes it difficult to incorporate domain specific constraints at the transceiver. For instance, a common transmit constraint is the power budget which in most practical cases would be different for different users. Thus designing a transmission scheme with per user power constraints becomes important, and can be achieved using the tensor framework as it maintains the identifiability of domains. Such a consideration has been presented in \cite{MDPIpaper, 9737095}.

\subsubsection{Other Examples}
In \cite{AdithyaPaper}, several multi-domain communication systems such as MIMO OFDM, GFDM, FBMC are represented using the tensor contracted convolution and contracted product. Also \cite{AdithyaPaper} develops tensor based receiver equalization methods which are used to combat interference in communication systems using the notion of tensor inversion. A tensor based receiver and precoder are presented in \cite{ICCpaper} for a MIMO GFDM system where the channel is represented as a sixth order tensor. The tensor SVD and EVD presented in this paper is used to design transmit coding operations and perform an information theoretic analysis of the tensor channel \cite{9737095} leading to the notion of multi-domain water filling power allocation method. Also, the discrete multi-linear contracted convolution is used to design Tensor Partial Response Signaling (TPRS) systems to shape the spectrum and cross spectrum of transmit signals \cite{MDPIpaper}. The tensor inversion method can also be used to develop estimation techniques for various signal processing applications such as in big data or wireless communications as shown in \cite{MMSEJournal}.  Another example of use of tensor Einstein product is for Image restoration and reconstruction application where objective is to retrieve an image affected by noise, a focal-field distribution, and aperture function \cite{LuEVD}. The image data is stored as a three dimensional tensor and an order 6 tensor acts as a channel obtained from the point spread function \cite{LuEVD}, such that output is given using the Einstein product between input and channel. Another area where the Einstein product properties have been used is the multi-linear dynamical system theory \cite{MLTI1,MLTI2}. In \cite{MLTI1}, a generalized multi-linear time invariant system theory is developed using the Einstein product which can be applied for dynamical systems such as human genome, social networks, cognitive sciences. The notion of tensor Eigenvalue decomposition presented in this paper is used in \cite{MLTI1} to derive conditions for stability, reachability, and observability for dynamical systems. The Einstein product has this distinct advantage that it lets us develop tensor algebra notions similar to linear algebra at the same time without disturbing or reshaping the structure of tensors. In addition, the more general tensor contracted product and contracted convolution can be used to model multi-domain systems with any mode ordering as well.  

\section{Summary and Concluding Remarks}\label{Sec6}
This paper presented a review of tensor algebra concepts developed using the contracted product, more specifically the Einstein Product, extending the common notions in Linear Algebra to a multi-linear setting. In particular, the notion of tensor inverse, singular and Eigenvalues decompositions, LU decomposition were discussed. We also studied the tensor networks representations of tensor contractions and convolutions. The notions of time invariant discrete and continuous multi-linear systems which can be defined using the contracted convolutions were also presented. We presented an application in a multi-domain communication system where the channel is modelled as a multi-linear system. The multi-linearity of the channel allowed us to develop a receiver which jointly combats interference across all the domains, thereby giving much better BER and MSE performance as compared to linear receivers which act on a specific domain at a time. The tensor algebra notions discussed in this paper has extensive applications in various fields such as Communications, signals and systems, controls, image processing, to name  a few. In the presence of several other tensor tutorial papers in literature, this paper by no means intends to summarize all the multi-linear algebra concepts, but provides a tutorial style introduction to the main concepts from a signals and systems perspective.   
\appendices

\end{spacing}

\bibliography{AllRef}
\bibliographystyle{IEEEtran}

\end{document}